\documentclass[11pt]{article}

\UseRawInputEncoding

\usepackage{amsmath}
\usepackage{amssymb}
\usepackage{amsthm}
\usepackage{url}
\usepackage{tikz}
\usepackage{listings}
\usepackage{pdfpages}
\usepackage{ascmac}
\usepackage[top=20truemm, bottom=20truemm]{geometry}

\usepackage{ulem}

\newtheorem{thm}{Theorem}[section]
\newtheorem{prop}[thm]{Proposition}
\newtheorem{cor}[thm]{Corollary}
\newtheorem{lem}[thm]{Lemma}

\newtheorem{pf}{Proof}

\newtheorem{defn}[thm]{Definition}
\newtheorem{definition}{Definition}

\usepackage{amssymb}
\usepackage{latexsym}

\newcommand{\C}{\mathbb C}

\newcommand{\R}{\mathbb R}
\newcommand{\Z}{\mathbb Z}

\newcommand{\pfqed}{\begin{flushright}Q.E.D.\end{flushright}}

\newcommand{\stkap}{*_{\kappa}} 
\title{
{\bf \LARGE Augmented Star Products and their Applications %to \\ Quantum Physics
}
}

\author{
Naoya MIYAZAKI
\\
Department of Mathematics\\Hiyoshi Campus
\\
Keio University
\\
Yokohama, 223-8521, JAPAN}
\date{August/05th/2026}

\begin{document}
\maketitle

\bigskip
\par\medskip\noindent{\small {\bf Abstract}:
In this paper, we introduce a generalized framework for star products that preserve associativity while not necessarily obeying the canonical commutation relations. Within this framework, we formulate the augmented star product and investigate the associated augmented star exponentials. We further demonstrate several applications of these constructions to special functions and to problems arising in quantum physics.

\par\medskip\noindent
\noindent{\bf Mathematics Subject Classification (2010):} Primary 58B32; 
Secondary 53C28, 53D55
\par\medskip\noindent
\noindent{\bf Keywords:} deformation quantization, star product, star exponential, special functions
%Riccati equations, the Baker-Campbell-Hausdorff formula  
etc. 
\par\medskip\noindent

\newpage

{\small 
\tableofcontents
%\clearpage
}

\bigskip\bigskip\bigskip

\newpage

\section{Introduction} 
Deformation quantization or star product introduced in \cite{bffls} is known as a general method for quantization for an arbitraty Poisson manifold \cite{dl, ds, dt, fedosov, gutt, kontsevich, marcolli-penrose, miyazaki00, miyazaki05, miyazaki07IJGMMP, miyazaki07PM, ommy97banach, ommy98jms, ommy98cmp, sternheimer, waldmann, y02, y17}. 
As to this method, we studied star exponentials, star functions, ordering problem, convergence problem in \cite{mmoy, maillard, ommy02contempmath, ommy03jlt, ommy05noncommgeomphys, ommy07yukawa, 
ommy07contenmpmath, ommy07lmp, ommy08asterisque, ommy14msri}, etc.

The notion of a star product is closely related to the canonical commutation relation in quantum mechanics. This relation is expressed by a pair of operators $\hat{q}$ and $\hat{p}$ satisfying
\[
[\hat{q}, \hat{p}]
:=\hat{q}\hat{p}-\hat{p}\hat{q}=\sqrt{-1}\hbar 
\qquad {\rm (CCR)}\]
where $''\hat{q}=q\times''$  is a multiplication operator and $''\hat{p}=-\sqrt{-1}\hbar \partial_q''$ is a differential operator acting on the functions of $q$, and $\hbar$ is equal to the Planck constant divided by $2\pi$. The algebra generated by $\hat{q}, \hat{p}$ with ${\rm (CCR)}$ is called the Weyl algebra. 
Beyond the operator formalism, there is 
another way to obtain the same algebraic structure. 
In \cite{bffls}, the authors introduced an 
associative but noncommutative deformation 
of the ordinary product of functions, 
directed by the symplectic structure of phase space. This construction is now called deformation quantization, 
and the resulting deformed product is 
referred to as a star product.

The typical example of star product is the Moyal product given as follows: For functions $f,~g$ of $(q,p)=(q_1,\ldots, q_n,p_1,\ldots,p_n)$ on a symplectic space,   
\[
\begin{array}{lll}
\vspace{2mm}
\displaystyle
f*_Mg&=&
\displaystyle
f\left[\exp\frac{\sqrt{-1}\hbar}{2}\sum_{i}
\left(\overleftarrow{\partial}_{q_i} \cdot \overrightarrow{\partial}_{p_i}
-\overleftarrow{\partial}_{p_i} \cdot \overrightarrow{\partial}_{q_i}\right)\right]g \\
\vspace{2mm}
\displaystyle
&=&
\displaystyle
f\sum_{\ell=0}^{\infty}\left(\frac{\sqrt{-1}\hbar}{2}
\Bigr(
\sum_{i}
\overleftarrow{\partial}_{q_i} \cdot \overrightarrow{\partial}_{p_i}
-\overleftarrow{\partial}_{p_i} \cdot \overrightarrow{\partial}_{q_i}\Bigr)\right)^{\ell}
g
\end{array}
\]
where we use a Poisson bracket (biderivation) defined by 
\[
\{f,g\}
=f\Bigr(\sum_{i}
\overleftarrow{\partial}_{q_i} \cdot \overrightarrow{\partial}_{p_i}
-\overleftarrow{\partial}_{p_i} \cdot \overrightarrow{\partial}_{q_i}\Bigr)g
:= \sum_i\left(\partial_{q_i} f\cdot\partial_{p_i} g-
\partial_{p_i} f\cdot\partial_{q_i} g\right). 
\]
Now we easily have 
\[
\begin{array}{lll}
\vspace{2mm}
\displaystyle
q_j*_M p_k&=&
\displaystyle
q_j\sum_{\ell=0}^{\infty}\left(\frac{\sqrt{-1}\hbar}{2}
\Bigr(
\sum_{i}
\overleftarrow{\partial}_{q_i} \cdot \overrightarrow{\partial}_{p_i}
-\overleftarrow{\partial}_{p_i} \cdot \overrightarrow{\partial}_{q_i}\Bigr)\right)^{\ell}
p_k
\\
\vspace{2mm}
&=&
\displaystyle
q_jp_k+\frac{\sqrt{-1}\hbar}{2}q_j\Bigr(
\sum_{i}
\overleftarrow{\partial}_{q_i} \cdot \overrightarrow{\partial}_{p_i}
-\overleftarrow{\partial}_{p_i} \cdot \overrightarrow{\partial}_{q_i}\Bigr)
p_k
%\\
%\vspace{2mm}
%&=&\displaystyle
=q_jp_k+\frac{\sqrt{-1}\hbar}{2} \delta_{jk}, 
\end{array}
\]
\[\begin{array}{lll}
\vspace{2mm}
\displaystyle
p_k*_M q_j
&=&
\displaystyle
p_k q_j+\frac{\sqrt{-1}\hbar}{2} 
p_k\Bigr(
\sum_{i}
\overleftarrow{\partial}_{q_i} \cdot \overrightarrow{\partial}_{p_i}
-\overleftarrow{\partial}_{p_i} \cdot \overrightarrow{\partial}_{q_i}\Bigr)
q_j
=p_kq_j-\frac{\sqrt{-1}\hbar}{2}\delta_{jk}.
\end{array}
\]
Thus we see $*_M$ satisfies (CCR):  
\[
[q_j,p_k]=q_j*_M p_k-p_k*_M q_j=\sqrt{-1}\hbar \delta_{jk}.
\]
For the Moyal star product, star exponentials were introduced and explicitly computed in \cite{bffls} for certain homogeneous polynomials, with the aim of determining the spectra of the corresponding observables.  
In contrast, in quantum field theory -specifically in the quantization of the free scalar field Hamiltonian-
one employs a different star product, known as the {\it normal star product}, which corresponds to normal ordering.  
Furthermore, based on the matrix Riccati equation, Maillard \cite{maillard} established a method for computing star exponentials on symplectic phase space for all polynomials of degree at most two, in any chosen ordering.

\par\medskip

The purpose of this paper is to introduce the notion of an {\it augmented star product}, which generalizes the classical concept of a star product, and to investigate its fundamental properties.  
By introducing {\it twisted Cayley transformations}, we derive explicit formulas for {\it augmented star exponentials} in arbitrary orderings.  
We also present several general deformation functions constructed via the augmented star product.
We also give applications of augmented star exponentials to special functions and quantum physics.

\section{Star product}
%\subsection{}
Let $({\mathbb C}^{2n}, \omega)$ be a complex phase space endowed with a symplectic form $\omega$,
whose points are denoted by
\[
{\mathbf x}=\sum_{j=1}^n q_je_j+\sum_{j=1}^n p_je_{n+j}=\sum_{j=1}^{2n}x_je_j
\]
where $(e_1,\ldots,e_n,e_{n+1},\ldots,e_{2n})$ is a symplectic basis, that is, 
\[
\omega(e_i, e_j)=0,~\omega(e_{n+i}, e_{n+j})=0,~\omega(e_i, e_{n+j})=\delta_{ij}
\]
for $i,j=1,\ldots,n$. 
The Poisson bracket of two elements $f, g \in C^{\infty}({\mathbb R}^{2n})$ is given by the formula 
\[
\{f,g\} 
= \sum_{j,k=1}^{2n} \Lambda^{jk}_0
     \frac{\partial f }{\partial x_j}\frac{\partial g}{\partial x_k}
\]
where the matrix $\Lambda_0$ is defined by $\Lambda^{j, n+k}_0
=-\Lambda^{n+k, j}_0=\delta_{jk}$, i.e., 
$\Lambda_0=\left(\begin{array}{cc}0 & I_n \\ -I_n & 0\end{array}\right)$.

\medskip
A star product of functions $f$ and $g$ is given by the formal expansion
\[
f*_{\Gamma}g=fg+\sum_{\ell=1}^{\infty} \frac{1}{\ell!}\left(\frac{i\hbar}{2}\right)^{\ell}C_{\ell}^{\Gamma}(f,g), 
\]
where the bi-differential operators $C_{\ell}^{\Gamma}$ are defined by 
\[
C_{\ell}^{\Gamma}(f,g):=\sum_{j_r,k_s=1}^{2n} \Gamma^{j_1k_1}\cdots\Gamma^{j_{\ell}k_{\ell}}\partial_{j_1\ldots j_{\ell}}f \cdot \partial_{k_1\ldots k_{\ell}}g, 
\]
with the coefficients $\Gamma^{jk}$ of arbitrary constant complex numbers and 
\[
\partial_{j_1\ldots j_{\ell}}f=
\frac{\partial^{\ell}}{\partial x_{j_1}\cdots \partial x_{j_\ell}}f. 
\]
Note that we sometimes use Einstein's rule for $j_1,\cdots, j_\ell, k_1, \cdots, k_\ell.$

\medskip
\noindent
Furthermore, we suppose that these bi-differential operators satisfy the following relations:
\begin{equation}\label{associativity}
\sum_{r+s=\ell}\frac{1}{r!s!}C_r^{\Gamma}(C_s^{\Gamma}(f,g),h)
=
\sum_{r+s=\ell}\frac{1}{r!s!}C_r^{\Gamma}(f,C_s^{\Gamma}(g,h)),
\end{equation}
\begin{equation}\label{commutationrelationrule}
C_1^{\Gamma}(f,g)-C_1^{\Gamma}(g,f)=2\{f,g\}
\end{equation}
Relation (\ref{associativity}), which implies the associativity of the star product.
Then, relation (\ref{commutationrelationrule}), which implies the relation
$q_j*p_k-p_k*q_j=i\hbar \delta_{jk}$, 
is equivalent to the following formula:
\begin{equation}\label{ccr}
\Gamma-{}^t\Gamma=2\Lambda_0
\end{equation}
where ${}^t\Gamma$ denotes the transpose of the matrix $\Gamma$. We call this matrix $\Gamma$ the {\bf ordering matrix}. 
As we can easily see, $\Gamma=\left(\begin{array}{cc}0 & I_n \\ -I_n & 0\end{array}\right)$, $\Gamma=\left(\begin{array}{cc}0 & 2I_n \\ 0 & 0\end{array}\right)$ and $\Gamma=\left(\begin{array}{cc}0 & 0 \\ -2I_n & 0\end{array}\right)$ give the standard ordering, the Weyl ordering and the anti-standard ordering star product respectively.

\begin{prop}
For any ordering matrix $\Gamma$ satisfying 
$\Gamma-{}^t\Gamma=2\Lambda_0$, $*_\Gamma$ gives a star product. 
\end{prop}

\section{Generalization of star product}
Up to now, we mainly concerned with star products which appear in deformation quantization. Hereafter, extending the notion of star product, we introduce {\it augmented star  product (AS-product for short)}. 

\subsection{Definition and fundamental properties}
First we give the definition:
\begin{defn}
For any functions $f,g$ and arbitrary complex $N\times N$-matrix\footnote{The number $N$ is not necessarily even. }
$\Gamma\in M_N({\mathbb C})$, {\bf augmented  star product} $*_{\Gamma}$ is defined as follows:
\[
f*_{\Gamma}g=fg+\sum_{\ell=1}^{\infty} \frac{1}{\ell!}\left(\frac{i\hbar}{2}\right)^{\ell}C_{\ell}^{\Gamma}(f,g), 
\]
where the bi-differential operators $C_{\ell}^{\Gamma}$ are defined by 
\[
C_{\ell}^{\Gamma}(f,g)
:=\sum_{
\tiny
\begin{array}{c}
j_r,k_s=1\\
(r,s=1,\ldots, \ell)
\end{array}
}^{N} \Gamma^{j_1k_1}\cdots
\Gamma^{j_{\ell}k_{\ell}}\partial_{j_1\ldots j_{\ell}}f \cdot \partial_{k_1\ldots k_{\ell}}g.  
\]
Note the coefficients $\Gamma^{jk}$ are arbitrary constant complex numbers and
\[
\partial_{j_1\ldots j_{\ell}}f=
\frac{\partial^{\ell}}{\partial x_{j_1}\cdots \partial x_{j_\ell}}f.
\]

\noindent
We set 
\[
\begin{array}{l}
\vspace{2mm}
\displaystyle 
\Lambda =\frac{1}{2}(\Gamma-{}^t\Gamma)\quad\mbox{(skew part)} ,
\\
\vspace{2mm}
\displaystyle 
K=\frac{1}{2}(\Gamma+{}^t\Gamma)
\quad\mbox{(symmetric part)}.
\end{array}
\]
Note that $\Lambda$ is not necessarily  nondegenrate. 
\end{defn}
\noindent

\noindent
For augmented star products, we use the following notations:
\[
\begin{array}{ccl}
\Lambda=J&: &\mbox{ an original matrix (skew symmetric part)}, \\
K&: & \mbox{ an ordering parameter matrix (symmetric part)}, \\
\Gamma=\Lambda+ K&:& \mbox{ an ordering matrix} .
\end{array}
\]
When we consider an augmented star product, 
the skew symmetric part $\Lambda$ might be degenerate.  For example, when $\Gamma=\Lambda+K=0+K$, the augmented star product is commutative.  

%We sometimes call $\Lambda$ an {\bf origin}, 
%and call $\Gamma$ (resp. $K$) 
%an {\bf ordering matrix} 
%(resp. an {\bf ordering parameter matrix}). 

%\noindent
%We DO NOT suppose that these bi-differential %operators satisfy the following relation:
%\begin{equation}\label{poisson}
%C_1^{\Gamma}(f,g)-C_1^{\Gamma}(g,f)=2\{f,g\}.
%\end{equation}

\begin{prop}
\begin{enumerate}
\item For any ordering matrix $\Gamma$, the augmented star product gives an associative product. 
\item When a ordering matrix $\Gamma$ is symmetric (i.e. the skew part $\Lambda$ of $\Gamma$ is zero), the augmented star product is commutative. 
\end{enumerate}
\end{prop}

\begin{pf}
%We give a detailed proof of associativity. 
We use Einstein's convention. 

\noindent
Proof of (i). 
\begin{equation}\label{r-to-l}
\begin{array}{lll}
\vspace{3mm}
&&\displaystyle \sum_{r+s=N}\frac{1}{r!s!} C_r^{\Gamma} (f, C_s^{\Gamma} (g,h)) 
\\
\vspace{2mm}
&=&\displaystyle 
\sum_{r+s=N}\frac{1}{r!s!}  
\Gamma^{i_1 j_1} \cdots \Gamma^{i_r j_r} \partial_{i_1 \cdots i_r}f 
\\
\vspace{3mm}
&&
\displaystyle
\qquad\qquad\qquad
\cdot 
\Gamma^{k_1 \ell_1}\cdots \Gamma^{k_s \ell_s}
\left(
\sum_{a=0}^r \frac{r!}{a!(r-a)!} \partial_{j_1 \cdots j_a}
 \partial_{k_1 \cdots k_s}g \cdot \partial_{j_{a+1} \cdots j_r}
\partial_{\ell_1 \cdots \ell_s}h) \right) 
\\
\vspace{3mm}
&&(\mbox{Put }r=a+b.)
\\
\vspace{3mm}
&=&\displaystyle 
\sum_{a+b+s=N}\frac{1}{a!b!s!}  
\Gamma^{i_1 j_1} \cdots \Gamma^{i_{a+b} j_{a+b}} 
\Gamma^{k_1 \ell_1}\cdots \Gamma^{k_s \ell_s}
\partial_{i_1 \cdots i_{a+b}}f 
\\
\vspace{3mm}
&&
\displaystyle
\qquad\qquad\qquad
\cdot 
\left(
 \partial_{j_1 \cdots j_a}
 \partial_{k_1 \cdots k_s}g \cdot \partial_{j_{a+1} \cdots j_{a+b}}
\partial_{\ell_1 \cdots \ell_s}h) \right) 
\\
\vspace{3mm}
&&(\mbox{Replacing }k_1\to i_{a+b+1},\ldots,k_s\to i_{a+b+s}, \ell_1\to j_{a+b+1},\ldots,\ell_s\to j_{a+b+s}.)
\\
\vspace{3mm}
&=&\displaystyle 
\sum_{a+b+s=N}\frac{1}{a!b!s!}  
\Gamma^{i_1 j_1} \cdots \Gamma^{i_{a+b} j_{a+b}} 
\Gamma^{i_{a+b+1} j_{a+b+1}}\cdots \Gamma^{i_{a+b+s} j_{a+b+s}}\\
%\vspace{3mm}
&&\qquad\qquad \displaystyle 
\underbrace{\partial_{i_1 \cdots i_{a+b}}}_{a+b-times}f \cdot 
\underbrace{\partial_{i_{a+b+1} \cdots i_{a+b+s} }}_{s-times}  
\underbrace{\partial_{j_1 \cdots j_a}}_{a-times}
 g \cdot \underbrace{\partial_{j_{a+1} \cdots j_{a+b}}
\partial_{j_{a+b+1} \cdots j_{a+b+s} }}_{b+s-times}h). 
%\cdots\cdots(*) 
\end{array}
\end{equation}

Similarly, 
\begin{equation}\label{l-to-r}
\begin{array}{lll}
&&\displaystyle \sum_{r+s=N}\frac{1}{r!s!} C_s^{\Gamma} (C_r^{\Gamma} (f, g) ,h) 
\\
\\
\vspace{3mm}
&=&\displaystyle 
\sum_{a+b+s=N}\frac{1}{a!b!s!}
\Gamma^{i_1 j_1} \cdots \Gamma^{i_a j_a} 
\Gamma^{i_{a+1} j_{a+1} }\cdots \Gamma^{i_{a+b} j_{a+b} }\cdots \Gamma^{i_{a+b+1} j_{a+b+1} } \cdots \Gamma^{i_{a+b+s} j_{a+b+s} } 
\\
&&
\qquad\qquad\displaystyle 
\underbrace{ \partial_{i_1 \cdots i_a} \partial_{i_{a+1} \cdots i_{a+b} } }_{a+b-times} f 
\cdot \underbrace{ \partial_{i_{a+b+1} \cdots i_{a+b+s}} }_{s-times} 
\underbrace{ \partial_{j_1 \cdots j_a} }_{a-times}  g  \cdot \underbrace{ \partial_{j_{a+1} \cdots j_{a+b+s}} }_{b+s-times} h .  %~~\cdots\cdots(**)
\end{array}
\end{equation}
Thus, compairing \eqref{r-to-l} and \eqref{l-to-r}, 
we have associativity. %$L_t(f,g,h)=R_t(f,g,h)$. 

\noindent
Note that 
we only assume that $\Gamma$ is constant.  
%torsion-free and $\partial\Gamma=0$. 
We do not assume that $\Gamma=\overleftarrow{\partial_i}\Gamma^{ij}
\overrightarrow{\partial_j}$ is a Poisson bivector. 

\smallskip
\noindent
Proof of (ii) is obvious. 
%Thus we have $L_t(f,g,h)=R_t(f,g,h)$. 
\begin{flushright}Q.E.D.\end{flushright}

\end{pf}

\subsection{Eequivalence operators (intertwiners) for augmented star products}
In the present section, we study fundamental formulas of augmented star products and the equivalence operators relating augmented star products induced by ordering matrices.
Here after we use Einstien's convention. 

First, observing that
\begin{equation}\label{intertwiner-1}
\begin{array}{lll}
\vspace{2mm}
&&
\displaystyle 
\left(
\frac{i\hbar}{4} K^{ij} \partial_i\partial_j \right) (f\cdot g) \\
\vspace{2mm}
\displaystyle 
&=&\displaystyle \frac{i\hbar}{4} K^{ij} \partial_i\Bigr((\partial_j f)\cdot g + f\cdot (\partial_j g)\Bigr) \\
\vspace{2mm}
\displaystyle 
&=&\displaystyle 
\frac{i\hbar}{4} K^{ij} \Bigr(( \partial_i\partial_j f)\cdot g 
+ {(\partial_j f)\cdot ( \partial_i g)} 
+ ( \partial_i f)\cdot (\partial_j g) 
+ f\cdot (\partial_i\partial_j g) \Bigr) 
\\
\vspace{2mm}
\displaystyle 
&=&\displaystyle 
\frac{i\hbar}{4} K^{ij} ( \partial_i\partial_j f)\cdot g 
+ \frac{i\hbar}{4} ~ { {}^t K^{ij} (\partial_i f)\cdot ( \partial_j g)} 
+ \frac{i\hbar}{4} K^{ij} ( \partial_i f)\cdot (\partial_j g) 
+ \frac{i\hbar}{4} K^{ij}f\cdot (\partial_i\partial_j g) ,
\end{array}
\end{equation}
we introduce the folloinwg operators: 
\begin{defn}
\begin{equation}
\begin{array}{lll}
\vspace{2mm}
\displaystyle 
m(f\otimes g)&:=& f\cdot g , \\ 
\vspace{2mm}
\displaystyle 
L(f\otimes g) &:=&\displaystyle \frac{i\hbar}{4} K^{ij} ( \partial_i\partial_j f)\otimes g , \\ 
\vspace{2mm}
\displaystyle 
B(f\otimes g) &:=& \displaystyle \frac{i\hbar}{4} (K^{ij}+{}^t K^{ij}) (\partial_i f)\otimes ( \partial_j g) , \\
\vspace{2mm}
\displaystyle 
R(f\otimes g) &:=& \displaystyle \frac{i\hbar}{4} K^{ij} ( f \otimes \partial_i\partial_j g) , 
\end{array}
\end{equation}
where ${}^tK$ denotes the transpose of $K$.
\end{defn}
We have 
\begin{lem}
\[
{\rm (\ref{intertwiner-1})}
=m\circ( L + B + R) (f\otimes g).
\]
\end{lem}
\noindent
Note that  $L(f\otimes g)\not=L(g\otimes f)$. 
\par\medskip
We also have 
\begin{lem}\label{intertiwner-2}
\[
\begin{array}{lll}
\vspace{2mm}
(1)~L\circ B&=&B\circ L, \\ 
\vspace{2mm}
(2)~R\circ B&=&B\circ R, \\
(3)~L\circ R&=&R\circ L.
\end{array}
\]
Thus, we have 
\[
(L+B+R)^{\circ k}=\sum_{p+q+r=k} \frac{k!}{p!q!r!}L^{\circ p}\circ B^{\circ q}\circ R^{\circ r},
\]
where $L^{\circ p}$ denotes $\overbrace{L\circ \cdots \circ L}^{p-times}$, etc. 
\end{lem}
%\begin{pf}  
%Obvious. 
%First we show (1). 
%\[
%\begin{array}{lll}
%\vspace{2mm}
%L\circ B&=& \displaystyle L
%\Bigr(   \frac{i\hbar}{4} (K^{ij}+{}^t K^{ij})
% (\partial_i f)\otimes ( \partial_j g)  \Bigr) 
%\\
%\vspace{2mm}
%&=& 
%\displaystyle \frac{i\hbar}{4} K^{i_2j_2}
%\Bigr(   \frac{i\hbar}{4} 
%(K^{i_1j_1}+{}^t K^{i_1j}_1) 
%\partial_{i_2}\partial_{j_2}
%(\partial_{i_1} f) 
%\otimes ( \partial_{j_1} g)  \Bigr) . 
%\end{array}
%\]
%On the other hand, 
%\[
%\begin{array}{lll}
%\vspace{2mm}
%B\circ L&=&  \displaystyle B 
%\Bigr( \frac{i\hbar}{4} K^{i_2j_2} 
%(\partial_{i_2}\partial_{j_2} f)
%\otimes  g  \Bigr) 
%\\
%\vspace{2mm}
%&=& 
%\displaystyle 
%\Bigr(  \frac{i\hbar}{4} 
%(K^{i_1j_1}+{}^t K^{i_1j_1}) 
%\frac{i\hbar}{4} K^{i_2j_2} 
%(\partial_{i_1} \partial_{i_2}\partial_{j_2} f)
%\otimes  (\partial_{j_1} g ) \Bigr)  . 
%\end{array}
%\]
%Hence we have (1). 
%According to a similar manner, 
%we have (2) and (3). 
%\flushright Q.E.D.
%\end{pf}

\begin{lem}\label{intertwiner-lem}
Under the same notations and assumptions above, we have
\begin{equation}\label{intertwiner-k-power}
\begin{array}{lll}
\vspace{2mm}
&&
\displaystyle 
\frac{1}{k!}
\left(\frac{i\hbar}{4} K^{ij} \partial_i\partial_j \right)^k (f\cdot g) 
\\
\vspace{2mm}
\displaystyle 
&=&\displaystyle 
\sum_{p+q+r=k} \frac{1}{q!} \left(\frac{i\hbar}{4}\right)^q 
(K+{}^t K)^{i_1j_1} \cdots (K+{}^t K)^{i_qj_q}  \\
\vspace{2mm}
\displaystyle 
& &\displaystyle \qquad\quad  
\times \partial_{i_1\ldots i_q}
\frac{1}{p!} \left(\frac{i\hbar}{4}\right)^p 
K^{k_1\ell_1}\cdots K^{k_p\ell_p}\partial_{k_1\ldots k_p}\partial_{\ell_1\ldots \ell_p} f \\
\vspace{2mm}
\displaystyle 
&&\displaystyle \qquad\quad\qquad\quad  
\times \partial_{j_1\ldots j_q}
\frac{1}{r!} \left(\frac{i\hbar}{4}\right)^r 
K^{s_1r_1}\cdots K^{s_rt_r}\partial_{s_1\ldots s_r}\partial_{t_1\ldots t_r} g ,
\end{array}
\end{equation}
where $\partial_{i_1\ldots i_q}$ denotes 
$\partial_{i_1}\cdots \partial_{i_q}$, etc. 
and we apply Einstein's rule to the indices $i_1,\ldots, i_q,\ldots, t_1,\ldots,t_r$ etc. NOT to the indices $p,q,r$. 
\end{lem}

\begin{pf}
According to the previous notations and lemma, we have 
\[
\begin{array}{lcl}
\vspace{2mm}
\displaystyle 
\mbox{LHS of \eqref{intertwiner-k-power}}&=&\displaystyle \frac{1}{k!}m\circ (L+B+R)^{\circ k}(f\otimes g)\\
\vspace{2mm}
\displaystyle 
&\stackrel{ (\mbox{ \tiny Lemma~\ref{intertiwner-2} }) }{=}&\displaystyle \sum_{p+q+r=k}
\frac{1}{k!} ~ \frac{k!}{p!q!r!} ~ 
m\circ (L^{\circ p}\circ B^{\circ q}\circ R^{\circ r} (f\otimes g)) \\
%\qquad\qquad\qquad 
%(\mbox{Lemma~\ref{intertiwner-2} ) }\\ 
\vspace{2mm}
\displaystyle 
&=&\displaystyle 
\sum_{p+q+r=k} \frac{1}{q!} \left(\frac{i\hbar}{4}\right)^q 
(K+{}^t K)^{i_1j_1} \cdots (K+{}^t K)^{i_qj_q}  \\
\vspace{2mm}
\displaystyle 
& &\displaystyle \qquad\quad  
\times \partial_{i_1\ldots i_q}
\frac{1}{p!} \left(\frac{i\hbar}{4}K^{k\ell}\partial_k \partial_{\ell}\right)^p  f 
%\\
\vspace{2mm}
\displaystyle 
%&&\displaystyle \qquad\quad\qquad\quad  
\times \partial_{j_1\ldots j_q}
\frac{1}{r!} \left(\frac{i\hbar}{4}K^{k\ell}\partial_k \partial_{\ell}\right)^r 
 g 
\\
\vspace{2mm}
\displaystyle 
&=&\displaystyle 
\sum_{p+q+r=k} \frac{1}{q!} \left(\frac{i\hbar}{4}\right)^q 
(K+{}^t K)^{i_1j_1} \cdots (K+{}^t K)^{i_qj_q}  \\
\vspace{2mm}
\displaystyle 
& &\displaystyle \qquad\quad  
\times \partial_{i_1\ldots i_q}
\frac{1}{p!} \left(\frac{i\hbar}{4}\right)^p 
K^{k_1\ell_1}\cdots K^{k_p\ell_p}\partial_{k_1\ldots k_p}\partial_{\ell_1\ldots \ell_p} f \\
\vspace{2mm}
\displaystyle 
&&\displaystyle \qquad\quad\qquad\quad  
\times \partial_{j_1\ldots j_q}
\frac{1}{r!} \left(\frac{i\hbar}{4}\right)^r 
K^{s_1r_1}\cdots K^{s_rt_r}\partial_{s_1\ldots s_r}\partial_{t_1\ldots t_r} g 
\\
&=&
\mbox{RHS of \eqref{intertwiner-k-power}}. 
\end{array}
\]
%This completes the proof. 
\flushright Q.E.D. %$\Box$
\end{pf}

\par\medskip\noindent
Note that if $K={}^tK$ then 
$\displaystyle K^{ij}\partial_i \partial_j=
\frac{1}{2}(K+{}^t K)^{ij}\partial_i \partial_j$. 
%$K\overrightarrow{\partial}\overrightarrow{\partial}%=\frac{1}{2}(K+{}^t K)\overrightarrow{\partial}
%\overrightarrow{\partial}$. 
Hence we have 
\begin{cor}\label{intertwiner-cor}
Under the same notations and assumptions above, we have
\begin{equation}%\label{intertwiner-k-power}
\begin{array}{lll}
\vspace{2mm}
&&
\displaystyle 
\frac{1}{k!}
\left(\frac{i\hbar}{4} K^{ij} \partial_i\partial_j \right)^k (f\cdot g) 
\\
\vspace{2mm}
\displaystyle 
&=&\displaystyle 
\sum_{p+q+r=k} \frac{1}{q!} \left(\frac{i\hbar}{2}\right)^q 
K^{i_1j_1} \cdots K^{i_qj_q}  \\
\vspace{2mm}
\displaystyle 
& &\displaystyle \qquad\quad  
\times \partial_{i_1\ldots i_q}
\frac{1}{p!} \left(\frac{i\hbar}{4}\right)^p 
K^{k_1\ell_1}\cdots K^{k_p\ell_p}\partial_{k_1\ldots k_p}\partial_{\ell_1\ldots \ell_p} f \\
\vspace{2mm}
\displaystyle 
&&\displaystyle \qquad\quad\qquad\quad  
\times \partial_{j_1\ldots j_q}
\frac{1}{r!} \left(\frac{i\hbar}{4}\right)^r 
K^{s_1r_1}\cdots K^{s_rt_r}\partial_{s_1\ldots s_r}\partial_{t_1\ldots t_r} g . 
\end{array}
\end{equation}
%where $\partial_{i_1\ldots i_q}=\partial_{i_1}
%\cdots \partial_{i_q}$, etc. 
%and we apply Einstein's rule to the indices 
%$i_1,\ldots, i_q,\ldots, t_1,\ldots,t_r$ etc. NOT to the indices $p,q,r$. 
\end{cor}

\newpage\noindent
We also need the following lemma:

\begin{lem}\label{intertwinerlem2}
Assume that $K=(K^{ij})$ is a complex symmetric matrix. Then under the same notations above, we have 
\[
\begin{array}{l}
\vspace{2mm}

\begin{array}{lll}
\vspace{2mm}
(1)~&&\displaystyle 
\frac{1}{M!}f \Bigr((K+\Gamma)^{ij}\overleftarrow{\partial}_i \overrightarrow{\partial}_j\Bigr)^M g
\\
\vspace{2mm}
&=& 
\displaystyle 
\sum_{k+\ell=M}\frac{1}{k! \ell!} K^{i_1j_1}\cdots K^{i_kj_k}
\Gamma^{s_1t_1}\cdots \Gamma^{s_{\ell} t_{\ell}} 
(\partial_{i_1\ldots i_k}\partial_{s_1\ldots s_{\ell}} f) \times 
(\partial_{j_1\ldots j_k}\partial_{t_1\ldots t_{\ell}} g) , 
\end{array}
\\
\vspace{2mm}
\mbox{where }
f \Bigr((K+\Gamma)^{ij}\overleftarrow{\partial}_i \overrightarrow{\partial}_j\Bigr) g 
\mbox{ denotes }
 (K+\Gamma)^{ij} \partial_i f \partial_j g.
\\
\begin{array}{lll}
\vspace{2mm}
(2)~&&  \displaystyle
f*_{\Gamma,K} g \\
\vspace{2mm}
&:=& \displaystyle 
%f *_{\Gamma} 
%e^{\frac{i\hbar}{2}
%\overleftarrow{\partial_i}K^{ij}
%\overrightarrow{\partial j}}*_{\Gamma} g
\displaystyle \sum_{n=0}^{\infty} \frac{1}{n!}\left(\frac{i\hbar}{2}\right)^n K^{i_1j_1}\cdots K^{i_nj_n}(\partial_{i_1}\cdots \partial_{i_1} f )*_{\Gamma}(\partial_{j_1}\cdots \partial_{j_1} g )\\
&=&f*_{\Gamma + K} g . 
\end{array}

\end{array}
\]

\end{lem}

\begin{pf}
Note that we have 
\[
\begin{array}{lll}
\vspace{2mm}
&&f \Bigr((K+\Gamma)^{ij}\overleftarrow{\partial}_i \overrightarrow{\partial}_j\Bigr)^M g\\
\vspace{2mm}
&=& \displaystyle 
f \Bigr((K+\Gamma)^{i_1j_1}\overleftarrow{\partial}_{i_1} \overrightarrow{\partial}_{j_1}\Bigr)
\cdots 
\Bigr((K+\Gamma)^{i_Mj_M}\overleftarrow{\partial}_{i_M} \overrightarrow{\partial}_{j_M}\Bigr)g
\\
\vspace{2mm}
&=& 
\displaystyle 
\sum_{k+\ell=M}\frac{M!}{k! \ell!} K^{p_1q_1}\cdots K^{p_kq_k}
\Gamma^{s_1t_1}\cdots \Gamma^{s_{\ell} t_{\ell}} 
(\partial_{p_1\ldots p_k}\partial_{s_1\ldots s_{\ell}} f) \cdot 
(\partial_{q_1\ldots q_k}\partial_{t_1\ldots t_{\ell}} g) .
\end{array}
\] 
This shows the first assertion (1).

As for the second assertion, 
we see 
\[
\begin{array}{lll}
\vspace{2mm}
&&  
\displaystyle
f*_{\Gamma,K} g 
\\
%\vspace{2mm}
%&:=& 
%\displaystyle f *_{\Gamma} e^{\frac{i\hbar}{2}
%\overleftarrow{\partial_i}K^{ij}
%\overrightarrow{\partial j}}*_{\Gamma} g
%\\
\vspace{2mm}
&=& 
\displaystyle \sum_{n=0}^{\infty}\frac{1}{n!}\left(\frac{i\hbar}{2}\right)^n K^{i_1j_1}\cdots K^{i_nj_n}(\partial_{i_1}\cdots \partial_{i_1} f )*_{\Gamma}(\partial_{j_1}\cdots \partial_{j_1} g )\\
\vspace{2mm}
&=&
\displaystyle \sum_{n=0}^{\infty} \frac{1}{n!}\left(\frac{i\hbar}{2}\right)^n K^{i_1j_1}\cdots K^{i_nj_n}\\
\vspace{2mm}
&&\displaystyle \qquad\quad 
\times \sum_{k=0}^{\infty} \frac{1}{k!}\left(\frac{i\hbar}{2}\right)^k \Gamma^{s_1t_1}\cdots \Gamma^{s_kt_k}(\partial_{s_1}\cdots \partial_{s_k}\partial_{i_1}\cdots \partial_{i_n} f ) 
\cdot 
(\partial_{t_1}\cdots \partial_{t_k}\partial_{j_1}\cdots \partial_{j_n} g )
\\
\vspace{2mm}
&\stackrel{(1)}{=}&
f *_{\Gamma+K} g
\end{array}
\]
This completes the proof. \flushright Q.E.D.
\end{pf}

\noindent
We next show that 

\begin{thm}\label{intertwiner-thorem}
Assume that $K$ is a complex symmetric matrix. Then we have 
\[
e^{\frac{i\hbar}{4}K^{ij}{\overrightarrow{\partial}_i\overrightarrow{\partial}_j}}(f*_\Gamma g)
=\Bigr(e^{\frac{i\hbar}{4}K^{ij}{\overrightarrow{\partial}_i\overrightarrow{\partial}_j}}(f)\Bigr)*_{\Gamma+K}\Bigr( e^{\frac{i\hbar}{4}K^{ij}{\overrightarrow{\partial}_i\overrightarrow{\partial}_j}}(g)\Bigr), 
\]
where $K^{ij}{\overrightarrow{\partial}_i\overrightarrow{\partial}_j}f$ denotes 
$K^{ij}\partial_{ij}f$. 
That is, $e^{\frac{i\hbar}{4}K^{ij}{\overrightarrow{\partial}_i\overrightarrow{\partial}_j}}$ is an equivalence operator from the augmented star product $*_\Gamma$ to the product $*_{\Gamma+K}$. 
\end{thm}

\begin{pf} 
Thanks to Lemma~\ref{intertwiner-lem}, i.e.
\[ %\begin{equation}\label{intertwiner-k-power}
\begin{array}{lll}
\vspace{2mm}
&&
\displaystyle 
\uuline{
\frac{1}{k!}
\left(\frac{i\hbar}{4} K^{ij} \partial_i\partial_j \right)^k} (f\cdot g) 
\\
\vspace{2mm}
\displaystyle 
&=&\displaystyle 
\dashuline{
\sum_{p+q+r=k} \frac{1}{q!} \left(\frac{i\hbar}{4}\right)^q 
(K+{}^t K)^{i_1j_1} \cdots (K+{}^t K)^{i_qj_q} 
} \\
\vspace{2mm}
\displaystyle 
& &\displaystyle \qquad\quad  
\dashuline{ 
\times \partial_{i_1\ldots i_q}
\frac{1}{p!} \left(\frac{i\hbar}{4}\right)^p 
K^{k_1\ell_1}\cdots K^{k_p\ell_p}\partial_{k_1\ldots k_p}\partial_{\ell_1\ldots \ell_p} f 
}
\\
\displaystyle 
&&\displaystyle \qquad\quad\qquad\quad  
\dashuline{
\times \partial_{j_1\ldots j_q}
\frac{1}{r!} \left(\frac{i\hbar}{4}\right)^r 
K^{s_1r_1}\cdots K^{s_rt_r}\partial_{s_1\ldots s_r}\partial_{t_1\ldots t_r} g 
},
\end{array}
\]
%\end{equation}
we have 
\[\begin{array}{lcl}
\vspace{2mm}
&& \displaystyle e^{\frac{i\hbar}{4}K^{ij}{\overrightarrow{\partial}_i\overrightarrow{\partial}_j}}(f*_\Gamma g) \\
\vspace{2mm}
&=& 
\displaystyle e^{\frac{i\hbar}{4}K^{ij}{\overrightarrow{\partial}_i\overrightarrow{\partial}_j}}
\left( \sum_n\frac{1}{n!}\left(\frac{i\hbar}{2}\right)^n\Gamma^{I_1J_1}\cdots\Gamma^{I_nJ_n}\partial_{I_1\ldots I_n}f  \partial_{J_1\ldots J_n}g \right)
\\
\vspace{2mm}
&=&
\displaystyle 
\sum_k
\uuline{
\frac{1}{k!}
\left(\frac{i\hbar}{4} K^{ij} \partial_i\partial_j \right)^k} 
\left( \sum_n\frac{1}{n!}\left(\frac{i\hbar}{2}\right)^n\Gamma^{I_1J_1}\cdots\Gamma^{I_nJ_n}\partial_{I_1\ldots I_n}f  \partial_{J_1\ldots J_n}g \right)
\\
\vspace{2mm}
&=&
\displaystyle 
\sum_n\frac{1}{n!}\left(\frac{i\hbar}{2}\right)^n\Gamma^{I_1J_1}\cdots\Gamma^{I_nJ_n}
\dashuline{
\left[
\sum_k \sum_{p+q+r=k} \frac{1}{q!} \left(\frac{i\hbar}{4}\right)^q 
(K+{}^t K)^{i_1j_1} \cdots (K+{}^t K)^{i_qj_q} 
\right.} 
\\
\vspace{2mm}
\displaystyle 
& &\displaystyle \qquad 
\dashuline{
\left.
\times \partial_{i_1\ldots i_q}
\frac{1}{p!} \left(\frac{i\hbar}{4}K^{k\ell}\partial_k \partial_{\ell}\right)^p  \partial_{I_1\ldots I_n}f 
\times \partial_{j_1\ldots j_q}
\frac{1}{r!} \left(\frac{i\hbar}{4}K^{k\ell}\partial_k \partial_{\ell}\right)^r 
 \partial_{J_1\ldots J_n}g 
\right]} .
\end{array}
\]
Since we assume that ${{}^tK=K}$, we have 
\[
\begin{array}{lcl}
%\\
\vspace{2mm}
&{=}& 
\displaystyle 
\sum_n\frac{1}{n!}\left(\frac{i\hbar}{2}\right)^n\Gamma^{I_1J_1}\cdots\Gamma^{I_nJ_n}
\left[
\sum_k\sum_{p+q+r=k} \frac{1}{q!} \left(\frac{i\hbar}{2}\right)^q 
K^{i_1j_1} \cdots K^{i_qj_q}  
\right. 
\\
\vspace{2mm}
\displaystyle 
& &\displaystyle \qquad 
\times \partial_{i_1\ldots i_q}
\frac{1}{p!} \left(\frac{i\hbar}{4} \right)^p
K^{k_1\ell_1} \cdots K^{k_p\ell_p} 
\partial_{k_1\ldots k_p}\partial_{\ell_1\ldots\ell_p}  \partial_{I_1\ldots I_n}f 
\\
\vspace{2mm}
\displaystyle 
&&\displaystyle \qquad\qquad  
\left. 
\times \partial_{j_1\ldots j_q}
\frac{1}{r!} \left(\frac{i\hbar}{4} \right)^r
K^{s_1t_1} \cdots K^{s_rt_r} 
\partial_{s_1\ldots s_r}\partial_{t_1\ldots t_r}  \partial_{J_1\ldots J_n} g 
\right]
%\end{array}
%\]
\\
%\[
%\begin{array}{lcl}
\vspace{2mm}
%&=& 
%\displaystyle 
%\sum_{q} \frac{1}{q!} \left(\frac{i\hbar}{2}\right)^q 
%K^{i_1j_1} \cdots K^{i_qj_q}
%\left[\sum_n\frac{1}{n!}\left(\frac{i\hbar}{2}
%\right)^n\Gamma^{I_1J_1}\cdots\Gamma^{I_nJ_n}
%\right. 
%\\
%\vspace{2mm}
%\displaystyle 
%& &\displaystyle \qquad 
%\times \partial_{i_1\ldots i_q}
%\partial_{I_1\ldots I_n}
%\sum_p
%\frac{1}{p!} \left(\frac{i\hbar}{4} \right)^p
%K^{k_1\ell_1} \cdots K^{k_p\ell_p} 
%\partial_{k_1\ldots k_p}\partial_{\ell_1\ldots\ell_p} f 
%\\
%\vspace{2mm}
%\displaystyle 
%&&\displaystyle \qquad\qquad  
%\left. 
%\times \partial_{j_1\ldots j_q}
%\partial_{J_1\ldots J_n}
%\sum_r
%\frac{1}{r!} \left(\frac{i\hbar}{4} \right)^r
%K^{s_1t_1} \cdots K^{s_rt_r} 
%\partial_{s_1\ldots s_r}\partial_{t_1\ldots t_r}  g 
%\right]

\\
\vspace{2mm}
&~=& 
\displaystyle 
\sum_{q} \frac{1}{q!} \left(\frac{i\hbar}{2}\right)^q 
K^{i_1j_1} \cdots K^{i_qj_q}
\Bigr[\sum_n\frac{1}{n!}\left(\frac{i\hbar}{2}\right)^n\Gamma^{I_1J_1}\cdots\Gamma^{I_nJ_n}
%\right. 
\\
\vspace{3mm}
\displaystyle 
& &\displaystyle \qquad\qquad\qquad 
\times \partial_{i_1\ldots i_q}
\partial_{I_1\ldots I_n}
(e^{\frac{i\hbar}{4}K^{ij}\overrightarrow{\partial_i}\overrightarrow{\partial}_j} f) 
%\left. 
\times \partial_{j_1\ldots j_q}
\partial_{J_1\ldots J_n}
(e^{\frac{i\hbar}{4}K^{ij}\overrightarrow{\partial_i}\overrightarrow{\partial}_j}g)
\Bigr]
\\
\vspace{3mm}
\displaystyle 
&=&\displaystyle
(e^{\frac{i\hbar}{4}K^{ij}\overrightarrow{\partial_i}\overrightarrow{\partial}_j} f) 
*_{\Gamma,K}
(e^{\frac{i\hbar}{4}K^{ij}\overrightarrow{\partial}_i\overrightarrow{\partial}_j} g)

\\
\vspace{2mm}
&
\stackrel{\tiny \mbox{(Lemma~\ref{intertwinerlem2}(2))}}{=}
& \displaystyle 
(e^{\frac{i\hbar}{4}K^{ij}\overrightarrow{\partial_i}\overrightarrow{\partial}_j} f) 
*_{\Gamma+K}
(e^{\frac{i\hbar}{4}K^{ij}\overrightarrow{\partial_i}\overrightarrow{\partial}_j} g). 
\qquad\qquad
%\mbox{(Lemma~\ref{intertwinerlem2}(2))}. 
\end{array}
\]
This completes the proof. \flushright Q.E.D. 
%$\Box$

\end{pf}

\begin{defn}
Set $T^K:=e^{\frac{i\hbar}{4}K^{ij}{\overrightarrow{\partial}_i
\overrightarrow{\partial}_j}}$ 
and  $T^K(*_{\Gamma})
:=*_{\Gamma+K}=*_{\Gamma,K}$~
(cf. Lemma~\ref{intertwinerlem2}). 
\end{defn}
Then  
\[
e^{\frac{i\hbar}{4}K^{ij}{\overrightarrow{\partial}_i\overrightarrow{\partial}_j}}(f*_\Gamma g)
=\Bigr(e^{\frac{i\hbar}{4}K^{ij}{\overrightarrow{\partial}_i\overrightarrow{\partial}_j}}(f)\Bigr)
\uuline{ *_{\Gamma+K} }
\Bigr( e^{\frac{i\hbar}{4}K^{ij}{\overrightarrow{\partial}_i\overrightarrow{\partial}_j}}(g)\Bigr)
\]
is rewritten as 
\[
T^K(f*_{\Gamma} g)=(T^K(f))(\uuline{T^K(*_{\Gamma})})(T^K(g)). 
\]

\par\medskip\noindent
{\bf Remark.}  
Using the above equivalence operator 
$T^K=e^{\frac{i\hbar}{4}K^{ij}\partial_i\partial_j}$, 
we can represent several ordering in the following way:
\[
\begin{array}{c|c|c}
Quantization & K & \Gamma=\Lambda+K
\\
\hline
Weyl &  
\left[\begin{array}{cc} 0 & 0 \\ 0 & 0 \end{array}\right]
&
\left[\begin{array}{cc} 0 & 1 \\ -1 & 0 \end{array}\right]
\\
\hline
Standard & 
\left[\begin{array}{cc} 0 & 1 \\ 1 & 0 \end{array}\right]
&
\left[\begin{array}{cc} 0 & 2 \\ 0 & 0 \end{array}\right]
\\
\hline
Anti-Standard & 
\left[\begin{array}{cc} 0 & -1 \\ -1 & 0 \end{array}\right]
&
\left[\begin{array}{cc} 0 & 0 \\ -2 & 0 \end{array}\right]
\\
\hline
Normal & 
\left[\begin{array}{cc} \frac{-\hbar}{m\omega} & 0 \\ 0 & {-\hbar m\omega} \end{array}\right]
&
\left[\begin{array}{cc} \frac{-\hbar}{m\omega} & 1 \\ -1 & {-\hbar m\omega} \end{array}\right]
\\
\hline
Anti-Normal & 
\left[\begin{array}{cc} \frac{\hbar}{m\omega} & 0 \\ 0 & {\hbar m\omega} \end{array}\right]
&
\left[\begin{array}{cc} \frac{\hbar}{m\omega} & 1 \\ -1 & {\hbar m\omega} \end{array}\right]
\\
\hline
Husimi & 
\left[
\begin{array}{cc}
\frac{\eta}{i\hbar}\sigma^2 & 0 \\
0 & \frac{\eta}{i\hbar}\frac{1}{\sigma^2}
\end{array}
\right]
&
\left[
\begin{array}{cc}
\frac{\eta}{i\hbar}\sigma^2 & 1 \\
-1 & \frac{\eta}{i\hbar}\frac{1}{\sigma^2}
\end{array}
\right]
\\
\hline
damped 
&
\left[
\begin{array}{cc}
-2\gamma m & 0 \\
0 & \displaystyle 0
\end{array}
\right]
&
\left[
\begin{array}{cc}
-2\gamma m & 1 \\
-1 & \displaystyle 0
\end{array}
\right]
\end{array}
\]
where $m,~\omega,~\eta,~\sigma$ and $\gamma$ denote a mass parameter, an angular momentum parameter, a classical coarse-graining scale, a squeezing parameter and a damping parameter. 
For this table, see section 7. 

\section{Augmented star exponentials of linear forms}
In this section, we compute augmented  exponentials of linear forms, i.e., polynomials of order at most 1. 

%\subsection{Augmented eneral star exponentials 
%of linear forms}

\begin{prop}\label{maillard-deg1}
For a polynomial ${}^t {\bf b}{\bf x}+c,~{\bf b}\in {\mathbb C}^{2n}, c\in {\mathbb C}$, set 
\[
{\rm Exp}_{*_\Gamma}\frac{t}{i\hbar}({}^t {\bf b}{\bf x}+c) :=\sum_{k=0}^{\infty}\frac{1}{k!}
\overbrace{\frac{t}{i\hbar}({}^t {\bf b}{\bf x}+c){*_\Gamma}\cdots{*_\Gamma}\frac{t}{i\hbar}({}^t {\bf b}{\bf x}+c)}^{k-{times}}.
\]
Then we have 
\[
{\rm Exp}_{*_\Gamma}\frac{t}{i\hbar}({}^t {\bf b}{\bf x}+c) =\exp\left(\frac{t^2}{8i\hbar}{}^t{\bf b}(\Gamma + {}^t\Gamma){\bf b}  \right)
\times
\exp\left(\frac{t}{i\hbar}({}^t {\bf b}{\bf x}+c)\right).
\]

\end{prop}

\begin{pf}
Solving the defining equation  
\[
\begin{array}{lll}
\vspace{2mm}
\displaystyle
\frac{\partial F(t,{\bf x})}{\partial t}
&=&
\displaystyle
\frac{ \partial {\rm Exp}_{*_\Gamma}\frac{t}{i\hbar}({}^t {\bf b}{\bf x}+c)}{\partial t}\\
\vspace{2mm}
\displaystyle
&=&
\displaystyle
\frac{1}{i\hbar}({}^t {\bf b}{\bf x}+c){*_\Gamma}F(t,{\bf x})
\\
\vspace{2mm}
\displaystyle
&=&\displaystyle
\frac{1}{i\hbar}
\Bigr\{
({}^t {\bf b}{\bf x}+c)F(t,{\bf x})
+\frac{i}{2} \sum_{j,k=1}^{2n} b_j \Gamma^{jk} \partial_k F(t,{\bf x})
\Bigr\}. 
\end{array}
\]
with $F(0,{\bf x})=1$ and 
\[
{\rm Exp}_{*_\Gamma}\frac{t}{i\hbar}({}^t {\bf b}{\bf x}+c)
=F(t,{\bf x})=\frac{1}{f(t)}\exp i ({}^t {\bf k}(t){\bf x}+\ell(t))
\] 
shows that 
functions ${\bf k}(t), f(t)$ and $\ell(t)$ satisfy that 
\[
\begin{array}{l}
%\vspace{2mm}
\displaystyle
{\bf k}'(t)=\displaystyle-\frac{\bf b}{\hbar}, 
\qquad
\ell'(t)=\displaystyle-\frac{c}{\hbar}, 
\qquad
\frac{f'(t)}{f(t)}
=\displaystyle i \left( \frac{c}{\hbar}+\ell'(t)
-\frac{i}{2} {}^t{\bf k}(t) {}^t\Gamma {\bf b}\right)
\end{array}
\]
with 
${\bf k}(0)=0$ and $f(0)=1$. 
Thus we have 
\[
%\begin{array}{lll}
%\vspace{2mm}
\displaystyle
{\bf k}(t)=\displaystyle\frac{-1}{\hbar}t{\bf b},
\qquad
\ell(t)=\displaystyle\frac{-ct}{\hbar},
\qquad
f(t)=\displaystyle\exp\left(-\frac{t^2}{8i\hbar}{}^t{\bf b}(\Gamma + {}^t\Gamma){\bf b}  \right). 
%\end{array}
\]
This completes the proof. 
\flushright Q.E.D. 
\end{pf}
According to Baker-Campbell-Hausdorff formula, 
we also have the following formula, 
\begin{prop}\label{product-deg1}
For polynomials 
${}^t {\bf a}{\bf x}+c, {}^t {\bf b}{\bf x}+d,~{\bf a}, {\bf b}\in {\mathbb C}^{2n}, c,d \in {\mathbb C}$, 
we have 
\[
{\rm Exp}_{*_\Gamma}\frac{1}{i\hbar}({}^t {\bf a}{\bf x}+c)
*_{\Gamma} 
{\rm Exp}_{*_\Gamma}\frac{1}{i\hbar}({}^t {\bf b}{\bf x}+d)
=
{\rm Exp}_{*_\Gamma}\frac{1}{i\hbar}\left({}^t ({\bf a}+{\bf b}){\bf x}+\bigr(c+d+\frac{1}{8}[\![{\bf a},{\bf b}]\!] \bigr) \right)
\]
where 
$[\![{\bf a}, {\bf b}]\!]
:={}^t{\bf a}2\Lambda{\bf b}-{}^t{\bf b}2\Lambda{\bf a}$ 
and 
$\Lambda:=%\left(\begin{array}{cc}0 & I_n \\ -I_n & 0\end{array}\right)
\frac{1}{2}(\Gamma-{}^t\Gamma)$. Note that $\Lambda$ is not necessarily equal to $\Lambda_0:=\left(\begin{array}{cc}0 & I_n \\ -I_n & 0\end{array}\right)
$.  
%\bigr(\Lambda^{jk}\bigr)$. 
\end{prop}
Combining Proposition \ref{maillard-deg1} with Proposition \ref{product-deg1}, we have product formula of star exponentials: 
\begin{prop}\label{heisenberg-1}
\[
\begin{array}{lll}
\vspace{2mm}
&&
\displaystyle 
%{\rm Exp}_*\frac{1}{i\hbar}({}^t {\bf a}{\bf x}+c)
%*
%{\rm Exp}_*\frac{1}{i\hbar}({}^t {\bf b}{\bf x}+d) \\
\left[
\exp\left(\frac{1}{8i\hbar}{}^t{\bf a}(\Gamma + {}^t\Gamma){\bf a}  \right)
\times
\exp\left(\frac{1}{i\hbar}({}^t {\bf a}{\bf x}+c)\right)
\right]
\\
\vspace{2mm}
&&
\qquad\qquad 
\displaystyle
*_{\Gamma}
\left[\exp\left(\frac{1}{8i\hbar}{}^t{\bf b}(\Gamma + {}^t\Gamma){\bf b}  \right)
\times
\exp\left(\frac{1}{i\hbar}({}^t {\bf b}{\bf x}+d)\right)\right]\\
\vspace{2mm}
&=&
\displaystyle \exp\left(\frac{1}{8i\hbar}{}^t({\bf a+b})(\Gamma + {}^t\Gamma)({\bf a+b})  \right)
\times
\exp\left(\frac{1}{i\hbar}({}^t ({\bf a+b}){\bf x}+
(c+d+\frac{1}{8}[\![{\bf a}, {\bf b}]\!])\right).
\end{array}
\]
\end{prop}
Using an equivalence operator, we have 
\begin{prop}
\[
T^K\left[
\exp\left(\frac{1}{8i\hbar}{}^t{\bf a}(\Gamma + {}^t\Gamma){\bf a}  \right)
\times
\exp\left(\frac{1}{i\hbar}({}^t {\bf a}{\bf x}+c)\right)
\right]
\]
\[
%\qquad\qquad
=\left[
\exp\left(\frac{1}{8i\hbar}{}^t{\bf a}(\Gamma + {}^t\Gamma+2K){\bf a}  \right)
\times
\exp\left(\frac{1}{i\hbar}({}^t {\bf a}{\bf x}+c)\right)
\right].
\]
\end{prop}
\medskip

%\subsection{Adjoint actions of 
%general star exponentials}
Next we compute adjoint actions of 
augmented star exponentials of 
polynomials of order 2 to star exponentials 
of polynomials of order 1. 
\begin{prop}Let $A$ is 
a $N\times N$-symmetiric matrix with 
complex entries. 
Then we have
\[
\begin{array}{lll}
%\vspace{2mm}
\displaystyle
{\rm Ad}_{*_\Gamma}\left(
{\rm Exp}_{*_\Gamma}\left(\frac{{}^t{\bf x}A{\bf x}}{i\hbar}\right)\right)
{\rm Exp}_{*_\Gamma}\left(\frac{1}{i\hbar}{}^t {\bf b}{\bf x}\right)
%*{\rm Exp}_*\left({}^t{\bf x}A{\bf x}}{i\hbar}\right)^{-1}
&=&
\displaystyle
{\rm Exp}_{*_\Gamma}\left( \frac{2}{i\hbar} {}^t{\bf x}A\Lambda{\bf b} \right).
\end{array}
\]
\end{prop}

\begin{pf}
By a direct computation, we see 
\[
\begin{array}{lll}
\vspace{2mm}
&&\displaystyle
{\rm Ad}_{*_\Gamma}\left(
{\rm Exp}_{*_\Gamma}\left(\frac{{}^t{\bf x}A{\bf x}}{i\hbar}\right)\right)
\left({\rm Exp}_{*_\Gamma}\left(\frac{1}{i\hbar}{}^t {\bf b}{\bf x}\right)\right)
%\\
%\vspace{2mm}
%&=&\displaystyle
%{\rm Exp}_*\left( \frac{2}{i\hbar} {}^t{\bf x}A\Lambda{\bf b} \right)
\\
\vspace{2mm}
\displaystyle
&=&
\displaystyle
{\rm Exp}_{*_\Gamma}\left\{
\left( {\rm ad_{*_\Gamma}}\left(\frac{{}^t{\bf x}A{\bf x}}{i\hbar}\right)  \right) 
\left(\frac{1}{i\hbar}{}^t {\bf b}{\bf x}\right)
\right\}
\\
\vspace{2mm}
\displaystyle
&=&
\displaystyle
{\rm Exp}_{*_\Gamma}\left\{
\left(\frac{{}^t{\bf x}A{\bf x}}{i\hbar}\right) 
{*_\Gamma}\left(\frac{1}{i\hbar}{}^t {\bf b}{\bf x}\right)
- 
\left(\frac{1}{i\hbar}{}^t {\bf b}{\bf x}\right)
{*_\Gamma}\left(\frac{{}^t{\bf x}A{\bf x}}{i\hbar}\right)
\right\}
\\
%\end{array}
%\]
%\[
%\begin{array}{lll}
\vspace{2mm}
&=&\displaystyle
{\rm Exp}_{*_\Gamma}\frac{1}{i\hbar}\left({}^t{\bf x}A\Gamma{\bf b}-{}^t{\bf x}A{}^t\Gamma{\bf b}\right)
%\\
%\vspace{2mm}
=\displaystyle
{\rm Exp}_{*_\Gamma}\frac{2}{i\hbar}\left({}^t{\bf x}A\Lambda{\bf b}\right). 
%\qquad\qquad\qquad\qquad\qquad\qquad\Box
\end{array}
\]
\flushright Q.E.D. %$\Box$

\end{pf}

\bigskip
\section{Augmented star exponentials of quadratic forms}
In this section we are concerned with augmented star exponentials of quadratic forms. 

\medskip\noindent
{\bf Note that in this section, 
we assume that the original matrix $\Lambda$ satisfies} 
\[
-\Lambda={}^t\Lambda=\Lambda^{-1}.
\]
 
%\subsection{Twisted Cayley transformations and %augmented star exponentials of quadratic forms}
According to the defining equation of 
${\rm Exp}_{*_\Gamma}\left(
\frac{t}{i\hbar}Q\right)$ 
\begin{equation}\label{heisenberg-equation}
i\hbar\frac{\partial}{\partial t}
{\rm Exp}_{*_\Gamma}\left(\frac{t}{i\hbar}Q\right)
={Q}*_{\Gamma}{\rm Exp}_{*_\Gamma}\left(\frac{t}{i\hbar}Q\right)
\end{equation}
with a quadratic form $Q={}^txAx$ and  
\[
\displaystyle
{\rm Exp}_{*_\Gamma}\left(\frac{t}{i\hbar}Q\right)
=\frac{1}{f(t)}\exp(i[{}^txg(t)x]), %\qquad\quad(*)
\]
we have 
\[
g'(t)=-\frac{A}{\hbar} + A\Gamma g(t) 
+ {}^tg(t)~ {}^t\Gamma A - \hbar~ {}^t g(t) ~{}^t \Gamma A \Gamma g(t) 
\]
and 
\[
\frac{f'(t)}{f(t)}=\frac{\hbar}{2}tr \bigr(\tilde{A} g(t) \bigr), 
\]
where $\tilde{A}:={}^t\Gamma A \Gamma.$

Since $g=g(t)$ satisfies 
\[
g'=-\frac{A}{\hbar} + A\Gamma g 
+ {}^tg {}^t\Gamma A - \hbar {}^t g {}^t \Gamma A \Gamma g, 
\]
and $A$ is a symmetric matrix, 
taking the transpose, 
we 
have 
\[
\begin{array}{lll}
\vspace{2mm}
{}^tg'
&=&
\displaystyle
-\frac{{}^tA}{\hbar} + {}^tg {}^t\Gamma {}^tA 
+ {}^tA{}^{tt}\Gamma {}^{tt}g - \hbar {}^t g {}^t \Gamma {}^tA {}^{tt}\Gamma {}^{tt}g
\\
\vspace{2mm}
&=&
\displaystyle
-\frac{A}{\hbar} + A\Gamma g 
+ {}^tg {}^t\Gamma A - \hbar {}^t g {}^t \Gamma A \Gamma g. 
\end{array}
\]
Thus, if $g(0)={}^tg(0)$, we have $g(t)={}^tg(t)~(\forall t)$. 

We rewrite the first equation 
in the following way. 
%using $q$ instead of $g$. 
\begin{equation}
\begin{array}{lcl}
\vspace{2mm}
&&
\displaystyle
g'=-\frac{A}{\hbar} + A\Gamma g 
+ {}^tg {}^t\Gamma A - \hbar {}^t g {}^t \Gamma A \Gamma g
\\
\vspace{2mm}
&\Leftrightarrow&
\displaystyle
-\hbar g'(t)
=A + A\Gamma (-\hbar g) 
+ {}^t(-\hbar g) {}^t\Gamma A + {}^t (-\hbar g) {}^t \Gamma A \Gamma (-\hbar g)
\\
\vspace{2mm}
&&({\rm Setting~}q= - \hbar g) 
%\\
%\vspace{2mm}
%&\Leftrightarrow&
%q'=A+A\Gamma q + {}^t q{}^t\Gamma A + {}^tq
%{}^t\Gamma A \Gamma q
%=(1+{}^tq{}^t\Gamma) A (1+\Gamma g)
\\
\vspace{2mm}
&\Leftrightarrow&
q'=A+A(K+\Lambda) q + {}^t q (K-\Lambda)A + {}^tq (K-\Lambda) A (K+\Lambda) q .
\end{array} 
\end{equation}
Note if $A$ is symmetric then $q$ is also symmetric.  
\medskip

Next we introduce twisted Cayley transformation. 
\begin{defn}
\[
\begin{array}{lcl}
\vspace{2mm}
C^{K}(q)
&=&
\displaystyle
\Bigr( I-(I-\Lambda K)(-q\Lambda) \Bigr) \frac{1}{\Bigr(I+(I+\Lambda K)(-q\Lambda)\Bigr)}
\\
\vspace{2mm}
\displaystyle
&\stackrel{(*)}{=}&
\displaystyle
\frac{1}{\Bigr(I+(-q\Lambda)(I+\Lambda K)\Bigr)} \Bigr(I-(-q\Lambda)(I-\Lambda K)\Bigr).
\end{array}
\]
\end{defn}
Note that %We show $(*)$. 
\[
\begin{array}{lcl}
\vspace{2mm}
\displaystyle
&&
\{I+(-q\Lambda)(I+\Lambda K)\}\{I-(I-\Lambda K)(-q\Lambda)\}
\\
\vspace{2mm}
\displaystyle
&=&
I+(I-\Lambda K)q\Lambda-q\Lambda(I+\Lambda K)-q\Lambda(I+\Lambda K)(I-\Lambda K)(q\Lambda)
\\
\vspace{2mm}
\displaystyle
&=&
I-(\Lambda K)(q\Lambda)+(q\Lambda)(\Lambda K)-(q\Lambda)(I+\Lambda K)(I-\Lambda K)(q\Lambda).
\end{array}
\]
On the other hand, 
\[
\begin{array}{lcl}
\vspace{2mm}
\displaystyle
&&
\{I-(-q\Lambda)(I-\Lambda K)\}\{I+(I+\Lambda K)(-q\Lambda)\}
\\
\vspace{2mm}
\displaystyle
&=&
I-(\Lambda K)(q\Lambda)+(q\Lambda)(\Lambda K)-(q\Lambda)(I-\Lambda K)(I+\Lambda K)(q\Lambda).
\end{array}
\]
These prove $(*)$. 

\bigskip

According to  the twisted Cayley transformation, we have 
\begin{thm}
\[
C^K(q(t))'=2A\Lambda\cdot C^K(q(t)).
\]
Thus the solution of the equation above with the initial condition $C^K(q(0))=C_0$ is 
\[
C^K(q(t))=\exp (2A\Lambda t)C_0.
\]
If $q(0)=I$, $C^K(I)=I$, and then we have  
\[
C^K(q(t))=\exp(2A\Lambda t).
\]
\end{thm}

\begin{pf}
We put $q=q(t)$. 
\[
\begin{array}{lcl}
\vspace{2mm}
&&\displaystyle
\frac{d}{dt}(C^{K}(q(t)))
\\
\vspace{2mm}
&=&
\displaystyle
\left\{ 
\Bigr( I-(I-\Lambda K)(-q\Lambda ) \Bigr) \frac{1}{\Bigr(I+(I+\Lambda K)(-q\Lambda )\Bigr)}
\right\}'
\\
\vspace{2mm}
&&({\rm Using~}(A^{-1})'=-A^{-1}A'A^{-1})
\\
\vspace{2mm}
&=&
\displaystyle
(I-\Lambda K)(q'\Lambda ) { \frac{I}{I-(I+\Lambda K)(q\Lambda )} } 
\\
\vspace{2mm}
&&
\displaystyle
\qquad
+ (I+(I-\Lambda K)(q\Lambda ))
\frac{I}{I-(I+\Lambda K)(q\Lambda )}
\bigr\{ (-1)\bigr(-(I+\Lambda K)(q'\Lambda )\bigr) \bigr\}
{ \frac{I}{I-(I+\Lambda K)(q\Lambda )} }
%\\
%\vspace{2mm}
%&=&
%\displaystyle
%\Bigr\{
%\underline{ {\bf +}(I-\Lambda K)(q'\Lambda )} 
%
%+\uuline{(I+(I-\Lambda K)(q\Lambda ))
%\frac{I}{I-(I+\Lambda K)(q\Lambda )} } 
%({\bf +}(I+\Lambda K)(q'\Lambda ))
%\Bigr\}
%\boxed{ \frac{I}{(I-(I+\Lambda K)(q\Lambda )} } 
\\
\vspace{2mm}
&=&
\displaystyle
\Bigr\{
{ { \frac{I}{I-q\Lambda (I+\Lambda K)} } 
(I-q\Lambda (I+\Lambda K))(I-\Lambda K)q'\Lambda } 
\\
\vspace{2mm}
&&
\displaystyle
\qquad\qquad\qquad 
+
{ { \frac{I}{(I-(q\Lambda )(I+\Lambda K)} } 
(I+(q\Lambda )(I-\Lambda K))} 
(+(I+\Lambda K)(q'\Lambda ))
\Bigr\}
{ \frac{I}{(I-(I+\Lambda K)q\Lambda } } 
\end{array}
\]
\[
\begin{array}{lcl}
%\vspace{2mm}
%&=&
%\displaystyle
%\boxed{ \frac{I}{I-q\Lambda (I+\Lambda K)} }
%\Bigr\{
%(I-q\Lambda (I+\Lambda K)) 
%(I-\Lambda K)q'\Lambda  
%\\
%\vspace{2mm}
%&&
%\displaystyle
%\qquad\qquad\qquad 
%+
%(I+(q\Lambda )(I-\Lambda K))
%(+(I+\Lambda K)(q'\Lambda ))
%\Bigr\}
%{ \frac{I}{(I-(I+\Lambda K)q\Lambda } }
%\\
%\vspace{2mm}
%&=&
%\displaystyle
%{ \frac{I}{I-q\Lambda (I+\Lambda K)} }
%\Bigr\{
%(I-\Lambda K)-q\Lambda (I+\Lambda K)(I-\Lambda K)
%\\
%\vspace{2mm}
%&&
%\displaystyle
%\qquad\qquad\qquad 
%+
%(I+\Lambda K)+(q\Lambda )(I-\Lambda K)(I+\Lambda K)\Bigr\}(q'\Lambda )
%{ \frac{I}{(I-(I+\Lambda K)q\Lambda } }
%\\
\vspace{3mm}
&=&
\displaystyle
{ \frac{I}{I-q\Lambda (I+\Lambda K)} }
2(q'\Lambda )
{ \frac{I}{(I-(I+\Lambda K)q\Lambda } }
\\
\vspace{3mm}
&&
\displaystyle
(\mbox{Since~}q'=A+A(K+\Lambda ) q + {}^t q (K-\Lambda )A + {}^tq (K-\Lambda ) A (K+\Lambda ) q\mbox{ and }~{}^tq=q,   )
\\
\vspace{2mm}
&=&
\displaystyle
{ \frac{I}{I-q\Lambda (I+\Lambda K)} }
2
{ 
\Bigr\{
A+A(K+\Lambda ) q +  q (K-\Lambda )A + q (K-\Lambda ) A (K+\Lambda ) q
\Bigr\}
\Lambda 
}
{ \frac{I}{(I-(I+\Lambda K)q\Lambda } }
%\\
%\vspace{2mm}
%&=&
%\displaystyle
%{ \frac{I}{I-q\Lambda (I+\Lambda K)} }
%2
%{
%\{I-q\Lambda (I+\Lambda K)\}
%A\Lambda 
%}
%{ \frac{I}{(I-(I+\Lambda K)q\Lambda } }
\\
\vspace{2mm}
&=&
\displaystyle
2A\Lambda ~C^K(q(t)).
\end{array}
\]
If $C^K(q(0))=C_0$, 
\[
C^K(q(t))=\exp(2A\Lambda  t)C_0.
\]
If $q(0)=I$, then $C^K(I)=I$, and then we have  
\[
C^K(q(t))=\exp(2A\Lambda  t).
\]
\pfqed
\end{pf}

%\bigskip\bigskip\bigskip
%{\bf\large June 26thD'±'±'©'牺'Í$J$'ð$
%\Lambda$'ɏC³'µ'Ä'¢'È'¢}
%\bigskip\bigskip\bigskip

Note that the inverse of twisted Cayley transformation is as follows.
\begin{defn}
\[
(C^K)^{-1}(Y)
=
\frac{I}{I-\Lambda K+Y(I+\Lambda K)}(I-Y)\Lambda . 
\]
\end{defn}
\par\bigskip
Combining this formula with the above theorem, we also have 
\begin{thm}
Set 
\[
\displaystyle
\Pi=\frac{-1}{2}(\Gamma+{}^t\Gamma)
\Lambda=-K\Lambda , ~ 
\Gamma=\Lambda +K. 
\]
As for the phase part, it is a complex symmetric matrix and we obtain 
\[
\begin{array}{lcl}
\vspace{2mm}
q(t)
&=&
(C^K)^{-1}(~\exp(2A\Lambda  t)~)
\\
\vspace{2mm}
&=&
\displaystyle
\frac{I}{I-\Lambda K+Y(I+\Lambda K)}(I-Y)\Lambda \Bigr|_{Y=\exp(2A\Lambda  t)}
\\
\vspace{2mm}
&=&
\displaystyle
(I-Y)\frac{I}{I-\Lambda K+(I+\Lambda K)Y}\Lambda \Bigr|_{Y=\exp(2A\Lambda  t)}
\\
\vspace{2mm}
&=&
\displaystyle
(I-e^{2A\Lambda t})\frac{I}{I-\Lambda K+(I+\Lambda K)e^{2A\Lambda t}}\Lambda 
.
\end{array}
\]
And then  
\[
\begin{array}{lcl}
\vspace{2mm} 
g(t)
&=&
\displaystyle
-\frac{1}{\hbar} q(t)
\\
\vspace{2mm}
&=&
\displaystyle
\frac{-1}{\hbar}
\displaystyle
(I-e^{2A\Lambda t})\frac{I}{I-\Lambda K+(I+\Lambda K)e^{2A\Lambda t}}\Lambda 
\\
\vspace{2mm}
&=&
\displaystyle
\frac{-1}{\hbar}
\displaystyle
\frac{I}{I-\Lambda K+e^{2A\Lambda t}(I+\Lambda K)}(I-e^{2A\Lambda t})\Lambda 
\\
\vspace{2mm}
&=&
\displaystyle
\frac{1}{\hbar}
\Lambda\tanh(\Lambda At)
\{I-\Pi\tanh(t\Lambda A)\}^{-1}. 
\end{array}
\]
As for the amplitude, we obtain 
\[
\begin{array}{lcl}
\vspace{2mm}
f(t)&=&
\displaystyle
\det{}^{ \frac{1}{2} }
\left(
\frac{I-\Lambda K+e^{2A\Lambda t} (I+\Lambda K) }{2}
\right)
\\
\vspace{2mm}
&=&
\displaystyle
\exp\Bigr(\frac{t}{2}{\rm tr(A\Gamma)}\Bigr)
\times
\det{}^{ \frac{1}{2} }
\Bigr( 
\cosh (\Lambda At)-\Pi\sinh(\Lambda At)
\Bigr),  
\end{array}
\]
We choose for $\det{}^{\frac{1}{2}}$ 
the principal value . 
Summing up the above results, we obtain 
\begin{equation}
\begin{array}{lcl}
\vspace{2mm}
\displaystyle
{\rm Exp}_{*_\Gamma}\left(\frac{t}{i\hbar}Q\right)
&=&
\displaystyle
\frac{1}{f(t)}\exp(i[{}^txg(t)x]) 
\\
\vspace{2mm}
&=&
\displaystyle 
\det{}^{-\frac{1}{2} }
\left(
\frac{I-\Lambda K+e^{2A\Lambda t} (I+\Lambda K) }{2}
\right)
\\
\vspace{2mm}
&&
\displaystyle 
\qquad\qquad
\times
\exp
\left(
i {}^tx 
\frac{-1}{\hbar}
\Bigr[
\frac{I}{I-\Lambda K+e^{2A\Lambda t}(I+\Lambda K)}(I-e^{2A\Lambda t})\Lambda 
%(I-e^{2A\Lambda t})
%\frac{I}{I-\Lambda K+(I+\Lambda K)e^{2A\Lambda t}}\Lambda 
\Bigr]
x
\right).
\end{array}
\end{equation}

\end{thm}
\par\medskip\noindent
Note that the above result coincides with the result of Maillard \cite{maillard}. 
\par\medskip%\noindent
For quadratic forms $Q_i={}^tx A_i x~~(i=1,2)$, 
solving the initial-value problem 
\[
C^K(q(t))'=2A_1 \Lambda\cdot C^K(q(t))
\]
with the initial condition $C^K(q(0))=C_0 =\exp(2A_2\Lambda )$, we have  
\[
C^K(q(t))=\exp(2A_1\Lambda t)
\exp(2A_2\Lambda ) .
%C_0,\qquad 
%C^K(q(0))=C_0=\exp(2A_2\Lambda ) . 
\]
And then substituting the solution to the inverse of the Cayley transformation 
\[
(C^K)^{-1}(Y)
=
\frac{I}{I-\Lambda K+Y(I+\Lambda K)}(I-Y)\Lambda , 
\]
and hence we have: 
\begin{thm}
Under the same notations and assumptions above
, we have 
\[
\begin{array}{lcl}
\vspace{2mm}
\displaystyle
{\rm Exp}_{*_\Gamma}\left(
\frac{1}{i\hbar}Q_1\right)
\!\!\!
&{*_\Gamma}&
\!\!\!
\displaystyle
{\rm Exp}_{*_\Gamma}\left(
\frac{1}{i\hbar}Q_2\right)
\\
\vspace{2mm}
&=&
\displaystyle 
\det{}^{-\frac{1}{2} }
\left(
\frac{I-\Lambda K+e^{2A_1\Lambda }e^{2A_2\Lambda } (I+\Lambda K) }{2}
\right)
\\
\vspace{4mm}
&&
\displaystyle 
\qquad\quad
\times
\exp
\left(
i {}^tx 
\frac{-1}{\hbar}
\Bigr[
\frac{I}{I-\Lambda K+e^{2A_1\Lambda }e^{2A_2\Lambda }(I+\Lambda K)}(I-e^{2A_1\Lambda }e^{2A_2\Lambda })\Lambda 
%(I-e^{2A\Lambda })
%\frac{I}{I-\Lambda K+(I+\Lambda K)e^{2A\Lambda t}}\Lambda 
\Bigr]
x
\right)
\\
\vspace{2mm}
&=&
\displaystyle 
\det{}^{-\frac{1}{2} }
\left(
\frac{I-\Lambda K+
e^{ {\bf BCH}(2A_1\Lambda, 2A_2\Lambda) } 
(I+\Lambda K) }{2}
\right)
\\
\vspace{2mm}
&&
\displaystyle 
\qquad\quad
\times
\exp
\left(
{}^tx 
\frac{1}{ i \hbar}
\Bigr[
\frac{I}{I-\Lambda K
+e^{ {\bf BCH}(2A_1\Lambda, 2A_2\Lambda) }
(I+\Lambda K)}(I-
e^{ {\bf BCH}(2A_1\Lambda, 2A_2\Lambda) }  )\Lambda 
\Bigr]
x
\right)
.
\end{array}
\]
We choose for $\det{}^{\frac{1}{2}}$ 
the principal value. 
\end{thm}

\par\noindent
In the final equation, we used 
Baker-Campbell-Hausdorff formula\cite{yamanouchi-sugiura}:
\[
\begin{array}{lcl}
\vspace{2mm}
\exp(X)\exp(Y)
&=&
\exp( {\bf BCH}(X,Y) ), 
\\
\vspace{2mm}
{\bf BCH}(X,Y) 
&=&
\displaystyle 
\sum_{m=1}^{\infty}
\sum_{n=1}^{\infty}
\frac{1}{m}\frac{(-1)^{n-1}}{n}
\Bigr( \dashuline{{\bf BCH}_{m,n}(X,Y)}+(-1)^{m-1}
\uwave{{\bf BCH}_{m,n}(Y,X)} \Bigr)
\\
\vspace{2mm}
&=&
\displaystyle
\sum_{m=1}^{\infty}
\sum_{n=1}^{\infty}
\frac{1}{m}\frac{(-1)^{n-1}}{n}
\Biggr(~
\dashuline{ \sum_{\fbox{**}}
\frac{ ( {\rm ad}_X)^{p_1}  ({\rm ad}_Y)^{q_1} \cdots  ( {\rm ad}_X)^{p_{n-1}}  ({\rm ad}_Y)^{q_{n-1}}X }{p_1!q_1!\cdots \cdots p_{n-1}!q_{n-1}!} }
%\right.
\\
\vspace{2mm}
&&
\displaystyle 
\qquad\qquad\qquad 
\qquad\qquad\qquad 
+%\left.
(-1)^{m-1}
\uwave{
\sum_{\fbox{**}}
\frac{ ( {\rm ad}_Y)^{p_1}  ({\rm ad}_X)^{q_1} \cdots  ( {\rm ad}_Y)^{p_{n-1}}  ({\rm ad}_X)^{q_{n-1}}Y }{p_1!q_1!\cdots \cdots p_{n-1}!q_{n-1}!} } 
\Biggr), 
\end{array}
\]
where $\fbox{**}$ denotes 
\[
p_i,q_i\in \Z_{\geq 0 } ~~(i=1,\ldots ,n-1),~~~
p_i+q_i>0,~~~\sum_{i=1}^{n-1}(p_i+q_i)=m-1.
\]
Furthermore, when an open neighborhood $U$ of the identity of the symplectic group $Sp_{2n}(\C)$ is small enough, there exists 
an open neighborhood $V$ of
$0$ in  
the space $Sym_{2n}(\C)$ of symmetric matrices  such that 
\[
\frac{1}{2}\log(\cdot) \Lambda^{-1} : U \stackrel{\sim}{\longrightarrow} V,\qquad 
\exp(2(\cdot)\Lambda): V \stackrel{\sim}{\longrightarrow} U. 
\]  
Thus
\begin{thm}
Under the same notations and assumptions above, for matrices $g_1,~g_2$ in $U$ we have 
\[
\begin{array}{lcl}
\vspace{3mm}
&&
\displaystyle
{\rm Exp}_{*_\Gamma}\left(
\frac{1}{i\hbar}{}^tx \bigr( \frac{1}{2}(\log g_1)\Lambda^{-1} \bigr) x \right)
{*_\Gamma}~
\displaystyle
{\rm Exp}_{*_\Gamma}\left(
\frac{1}{i\hbar}{}^tx \bigr( \frac{1}{2}(\log g_2)\Lambda^{-1} \bigr) x  \right)
\\
\vspace{3mm}
&=&
\displaystyle 
\det{}^{-\frac{1}{2} }
\left(
\frac{I-\Lambda K+(g_1g_2) (I+\Lambda K) }{2}
\right)
%\\
%\vspace{4mm}
%&&
\displaystyle 
%\qquad\qquad
\cdot
\exp
\left(
 {}^tx 
\frac{1}{ i \hbar}
\Bigr[
\frac{I}{I-\Lambda K+(g_1g_2)
(I+\Lambda K)}(I-(g_1g_2))\Lambda 
\Bigr]
x
\right) . 
\end{array}
\] 
\end{thm}

\section{Applications to special functions}
In this section, we study modifications of special functions.

\subsection{Augmented  star Hermitian polynomials } 
Introducing the following product\footnote{It is known that usual commutative product $\cdot$ and convolution product give examples of commutative associative product. The product $*_{\stkap}$ also gives a commutative associative product. }, we can define augmented star Hermitian polynomials.

\begin{defn}
Set a $1\times 1$-ordering matrix with complex coefficient 
$\displaystyle \Gamma=J+K
=0+ \kappa \in \C$ and $\hbar =-i$. 
Using augmented star product 
\[
f*_{\Gamma}g=fg+\sum_{\ell=1}^{\infty} \frac{1}{\ell!}\left(\frac{i\hbar}{2}\right)^{\ell}C_{\ell}^{\Gamma}(f,g), 
\]
\[
C_{\ell}^{\Gamma}(f,g)
:=\sum_{
\tiny
\begin{array}{c}
j_r,k_s=1\\
(r,s=1,\ldots, \ell)
\end{array}
}^{N} \Gamma^{j_1k_1}\cdots
\Gamma^{j_{\ell}k_{\ell}}\partial_{j_1\ldots j_{\ell}}f \cdot \partial_{k_1\ldots k_{\ell}}g,   
\]
define
\[
f(z)*_{\kappa}g(z):=f(z)\exp\Bigr\{ \frac{1}{2}\overleftarrow{\partial_z}\kappa\overrightarrow{\partial_z}  \Bigr\}g(z).
\] 
for any functions $f(z), g(z)$ of one variable $z$ and a complex number $\kappa$. 
%2\pi i\kappa$. 
%\[
%\displaystyle \Gamma=2\pi i 
%K~(K\in Sym_{\mathbb C}(n)), 
%~~\hbar =-i , 
%\frac{-1}{2\pi}, 
%~~t=1, {\bf t}=\pi i {\bf b}, 
%~~{ z}={ x}. 
%\] 
\end{defn}
For example, we easily see that 
\[
\begin{array}{l}
\displaystyle
z \stkap z=z\cdot z+\frac{\kappa}{2}, \\
\displaystyle
z {\stkap} z^2=z^3 {\stkap}\frac{\kappa}{2}z=z^2 {\stkap} z,\\
\displaystyle
z {\stkap}z {\stkap}z=z_{\stkap}^{3}=z^3+\frac{3\kappa}{2}z.
\end{array}
\]
where $z^n$ denotes 
$\overbrace{z\cdots z}^{n-times}$,  
and $z_{\stkap}^{n}$ denotes 
$\overbrace{z {\stkap} 
\cdots  {\stkap}  z}^{n-times}.$
The product ${\stkap}$ gives an associative commutative product on the space of polynomials.  Thanks to Proposition \ref{maillard-deg1}, 
%i.e. 
%\[
%{\rm Exp}_{*_\Gamma}\frac{t}{i\hbar}({}^t 
%{\bf b}{\bf x}+c) 
%=\exp\left(\frac{t^2}{8i\hbar}{}^t{\bf b}
%(\Gamma + {}^t\Gamma){\bf b}  \right)
%\times
%\exp\left(\frac{t}{i\hbar}({}^t {\bf b}{\bf x}+c)\right).
%\]
we have  
\[
\begin{array}{lll}
\vspace{2mm}
%&&
\displaystyle\sum_{n=0}^{\infty}\frac{t^n}{n!}\overbrace{(2z)*_{\stkap} \cdots *_{\stkap} (2z)}^{n-times} 
%\\
%\vspace{2mm}
%&=&
={\rm Exp}_{*_{\stkap}}\bigr(2tz\bigr)\bigr|_{\kappa = -1}
%\\
%\vspace{2mm}
%&=&
=\displaystyle\exp\left(2tz-t^2\right)
=\displaystyle
\sum_{n=0}^{\infty}\frac{t^n}{n!}H_n(z),  
\end{array}
\]
%\[
%\begin{array}{lll}
%\vspace{2mm}
%&&
%\displaystyle\sum_{n=0}^{\infty}\frac{t^n}{n!}\overbrace{(\sqrt{2}x)*_{\stkap} \cdots *_{\stkap} (\sqrt{2}x)}^{n-times} 
%\\
%\vspace{2mm}
%&=&
%{\rm Exp}_{*_{\stkap}}\bigr(\sqrt{2}tx\bigr)\bigr|_{\stkap = -1}\\
%\vspace{2mm}
%&=&\displaystyle\exp\left(\sqrt{2}tx-\frac{1}{2}t^2\right)
%=
%\displaystyle
%\sum_{n=0}^{\infty}\frac{t^n}{n!}H_n(x),  
%\end{array}
%\]
where 
\[
H_n(z)=(-1)^n\exp(z^2)\frac{d^n\exp(-z^2)}{dz^n}
\]
is known as a Hermitian polynomial.  
Thus we also have 
\[
H_n(z)=(2z)_{{\stkap}}^n\Bigr|_{\kappa=-1}.
\]
Then we define {\bf star Hermitian polynomials}\footnote{By a similar manner, we can define star theta functions. See \cite{ommy14msri}. } $H_n(z,\kappa)$ as follows:
\begin{defn}
\[
\sum_{n=0}^\infty \frac{H_n(z,\kappa)}{n!}t^n
:=
{\rm Exp}_{\stkap}(2tz)=e^{\kappa t^2 +2tz}
=:
g_{\kappa}(t,z).
\]
Thus, we have $H_n(z,\kappa)=(2z)_{\stkap}^n.$
\end{defn}
Direct computation shows the following:
\begin{prop}
Under the above notations, we have 
\[
H_n(z,\kappa)
=
\sum_{j=0}^{\lfloor n/2 \rfloor}
\frac{n!}{j!(n-2j)!}
\kappa^j  (2z)^{n-2j} . 
\]
\end{prop}
\begin{pf}
\[
\begin{array}{lll}
\vspace{2mm}
\displaystyle 
&&
\displaystyle 
\sum_{n=0}^{\infty}
\frac{t^n}{n!}H_n(z,\kappa)
=g_{\kappa}(t,z)
\\
\vspace{2mm}
&=&
\displaystyle 
\exp(\kappa t^2 +2zt)
=\sum_{m=0}^{\infty}\frac{1}{m!}
(\kappa t +2zt)^m 
\\
\vspace{2mm}
&=&
\displaystyle 
\sum_{m=0}^{\infty}\frac{1}{m!}
\sum_{j=0}^{m}\frac{m!}{j!(m-j)!}
(\kappa t^2)^j(2zt)^{m-j}
\\
\vspace{2mm}
&=&
\displaystyle 
\sum_{m=0}^{\infty}
\sum_{j=0}^{m}\frac{1}{j!(m-j)!}
(\kappa t^2)^j(2zt)^{m-j}
\\
\vspace{2mm}
&&~~(\mbox{ Replacing }m+j=n)
\\
\vspace{2mm}
&=&
\displaystyle 
\sum_{n=0}^{\infty}
\frac{1}{n!}
\sum_{j=0}^{\lfloor n/2 \rfloor}
\frac{n!}{j!(n-2j)!}
\kappa^j t^n (2z)^{n-2j}.
\end{array}
\]
%\\
%\vspace{2mm}
%&=&
%=
%\displaystyle 
%e^{-\frac{1}{\kappa}z^2}
%\left( \kappa \frac{d}{dz} \right)^n 
%e^{\frac{1}{\kappa}z^2} \\
%\vspace{2mm}
%&=&\displaystyle (i\sqrt{\kappa})^n 
%\left[e^{x^2}%\left( \frac{d}{dx} \right)^n 
%e^{-x^2}\right]
%\Bigr|_{x=\frac{z}{i\sqrt{\kappa}}}  \\
%\vspace{2mm}
%&=&\displaystyle (i\sqrt{\kappa})^n
%\sum_{k=0}^{\lfloor n/2 \rfloor}
%\frac{(-1)^kn!}{k!(n-2k)!}
%(2 \frac{z}{i\sqrt{\kappa}} )^{n-2k}  \\
%\vspace{2mm}
%&=&\displaystyle 
%\sum_{k=0}^{\lfloor n/2 \rfloor}
%\frac{(-1)^kn!}{k!(n-2k)!}(2 z )^{n-2k}
%(i\sqrt{\kappa})^{2k}. 
%\end{array}
%\]
Thus, we have 
\[
H_n(z,\kappa)
=
\sum_{j=0}^{\lfloor n/2 \rfloor}
\frac{n!}{j!(n-2j)!}
\kappa^j  (2z)^{n-2j}
\]
\pfqed
\end{pf}

As for $H_n(z,\kappa)=(2z)_{\stkap}^n$, we have differential equations and recurrence relations in the following way:
\begin{prop}\label{zenkashiki}
Under the above definition and notations, we have a recurrence relation as follows:
\[
(2z)_{\stkap}^{n+1}-2z (2z)_{\stkap}^n -2\kappa n (2z)_{\stkap}^{n-1}=0,
\]
or equivalently, 
\[
H_{n+1}(z,\kappa)-2z H_n(z,\kappa)-2\kappa n H_{n-1}(z,\kappa)=0.
\]
\end{prop}
\noindent
{\bf Remark.} Note that 
\[
\begin{array}{c}
\vspace{2mm}
\displaystyle 
H_{n+1}(z,\kappa)-2z H_n(z,\kappa)-2\kappa n H_{n-1}(z,\kappa)=0
\end{array}
\]
converges the recurrence relation 
\[
\begin{array}{c}
\vspace{2mm}
\displaystyle 
H_{n+1}(z)-2z H_n(z)+2 n H_{n-1}(z)=0 
\end{array}
\]
when $\kappa \to -1$. 

\begin{pf}
Differentiate ${\rm Exp}_{\stkap}(2tz)$ by $t$, we have 
\[
\begin{array}{lll}
\vspace{2mm}
&&\displaystyle \frac{\partial}{\partial t}{\rm Exp}_{\stkap}(2tz) 
%=\displaystyle \frac{\partial}{\partial t} 
%\sum_{n=0}^\infty \frac{1}{n!}(2tz)_{\stkap}^n 
%\\
%\vspace{2mm}
=\displaystyle (2z) \stkap  {\rm Exp}_{\stkap}(2tz) \\
\vspace{2mm}
&=&\displaystyle (2z)  {\rm Exp}_{\stkap}(2tz) + \frac{\kappa}{2}\partial_z (2z) \cdot \partial_z  {\rm Exp}_{\stkap}(2tz) \\
\vspace{2mm}
&=&\displaystyle (2z)  {\rm Exp}_{\stkap}(2tz) + \kappa (2t {\rm Exp}_{\stkap}(2tz)),  
\end{array}
\]
where we use the definition of star product for the third equality, and the following formula for fourth equality:
\[
\displaystyle 
\frac{\partial}{\partial z} {\rm Exp}_{\stkap}(2tz) 
=
\frac{\partial}{\partial z} e^{\kappa t^2+2tz}=2t {\rm Exp}_{\stkap}(2tz). 
\]
Thus we see  
\[
%\frac{\partial}{\partial t}
\displaystyle 
\sum_{n=1}^\infty\frac{(2z)_{\stkap}^n}{(n-1)!}t^{n-1}
=
2z \sum_{n=0}^\infty\frac{(2z)_{\stkap}^n}{n!}t^{n} +\kappa 2t \sum_{n=0}^\infty\frac{(2z)_{\stkap}^n}{n!}t^{n}. 
\]
Comparing the coefficient of $t^n$, we have the first assertion.  
%\[
%(2z)_{\stkap}^{n+1}-2z (2z)_{\stkap}^n -2\kappa n (2z)_{\stkap}^{n-1}=0. 
%\]
Thanks to $H_n(z,\kappa)=(2z)_{\stkap}^n$, we see the second assertion. 

\flushright Q.E.D. 
\end{pf}
We also have 
\begin{prop}
Under the same notations above, we have a  descending operator as follows:
\[
\frac{d}{dz}H_n (z,\kappa) =2nH_{n-1}(z,\kappa)
\]
\end{prop}
\begin{pf}
\[
\frac{d}{dz}H_n (z,\kappa) 
=
\frac{d}{dz}(2z)_{*_{\kappa}}^n
=
2n(2z)_{*_{\kappa}}^{n-1}
=
2nH_{n-1}(z,\kappa)
\]

\pfqed
\end{pf}

Similarily, 
\begin{prop}\label{zenkashiki2}
Under the above definition and notations, we have a ascending operator as follows:
\[
\kappa 
\left( \frac{\partial}{\partial z} +2z\right) 
H_n(z,\kappa)=H_{n+1}(z,\kappa).
\]
or equivalently, 
\[
\kappa \frac{\partial}{\partial z}((2z)_{\stkap}^n) 
+2z (2z)_{\stkap}^n 
=(2z)_{\stkap}^{n+1},
\]
\end{prop}
\begin{pf}
Differentiate ${\rm Exp}_{\stkap}(2tz)$ by $t$, we have 
\[
\begin{array}{lll}
\vspace{2mm}
&&\displaystyle \frac{\partial}{\partial t}{\rm Exp}_{\stkap}(2tz) 
= 
\displaystyle (2z)  {\rm Exp}_{\stkap}(2tz) + \frac{\kappa}{2}\partial_z (2z) \cdot \partial_z  {\rm Exp}_{\stkap}(2tz) . 
\end{array}
\]
Replacing $\displaystyle {\rm Exp}_{\stkap}(2tz)=\sum_{n=0}^\infty \frac{t^n}{n!}(2z)_{\stkap}^n$, we see  
\[
\displaystyle 
\sum_{n=1}^\infty\frac{(2z)_{\stkap}^n}{(n-1)!}t^{n-1}
=
2z \sum_{n=0}^\infty\frac{(2z)_{\stkap}^n}{n!}t^{n} 
+\kappa \sum_{n=0}^\infty \frac{\partial}{\partial z}\frac{(2z)_{\stkap}^n}{n!}t^{n}. 
\]
Comparing the coefficient of $t^n$, we have the  assertion.  
%\[
%(2z)_{\stkap}^{n+1}-2z 
%(2z)_{\stkap}^n -2\kappa n 
%(2z)_{\stkap}^{n-1}=0. 
%\]
%Thanks to $H_n(z,\kappa)=(2z)_{\stkap}^n$, 
%we see the second assertion. 

\flushright Q.E.D.  %{$\Box$}
\end{pf}

Apply $\displaystyle \frac{d}{dz}$  
to the both sides of ascending operator 
\[
\kappa 
\left( \frac{\partial}{\partial z} +2z\right) 
H_n(z,\kappa)=H_{n+1}(z,\kappa).
\]
Then we have 
\[
\kappa\frac{d^1}{dz^2}H_n(z,\kappa)+
2z \frac{d}{dz} H_n(z,\kappa)+
2H_n(z,\kappa)
=\frac{d}{dz}H_{n+1}(z,\kappa). 
\]
Consider 
\[
\frac{d}{dz}H_n (z,\kappa) =2nH_{n-1}(z,\kappa)
\]
with $n\to n+1$, we have 
\[
\frac{d}{dz}H_{n+1} (z,\kappa) 
=2(n+1)H_{n}(z,\kappa). 
\]
Combining the above two formulas, we have the following differential equation. 
\begin{prop}Under the same notations above, we have the following differential equation: 
\[
\kappa\frac{d^2}{dz^2}H_n(z,\kappa)+
2z \frac{d}{dz} H_n(z,\kappa)+
2nH_n(z,\kappa)
=
0.
\]
\end{prop}

\medskip

Next we consider the orthogonality of $\{ H_n(z,\kappa) \}_{n=0,1,2,\ldots}$. Assume tnat ${\rm Re}(\kappa)<0, n\geq m$
\[
\begin{array}{lll}
\vspace{2mm}
&&
\displaystyle \int_{\mathbb R}e^{\frac{1}{\kappa}z^2}H_n(z,\kappa)H_m(z,\kappa)dz
=
\displaystyle \int_{\mathbb R}e^{\frac{1}{\kappa}z^2}
\left(e^{-\frac{1}{\kappa}z^2}
\bigr(\kappa\frac{d}{dz}\bigr)^n 
e^{\frac{1}{\kappa}z^2} \right) 
H_m(z,\kappa)dz \\
\vspace{2mm}
&=&
\displaystyle \int_{\mathbb R}
\left(
\bigr(\kappa\frac{d}{dz}\bigr)^n 
e^{\frac{1}{\kappa}z^2} \right) 
H_m(z,\kappa)dz 
=
\displaystyle \int_{\mathbb R}
e^{\frac{1}{\kappa}z^2}  
\bigr(-\kappa\frac{d}{dz}\bigr)^n H_m(z,\kappa)dz \\
\vspace{2mm}
&=&
\left\{
\begin{array}{ll}
\vspace{2mm}
0 & ~~(n>m) \\ 
\displaystyle (-\kappa)^2 2^n n!\sqrt{\pi}
\sqrt{-\kappa}%{\frac{\sqrt{\kappa}}{\sqrt{-1}} } 
& ~~(n=m) 
\end{array}
\right.
\end{array}
\]

%According to the identity
%\[
%\frac{d}{dt}{\rm Exp}_{*_\tau}(\sqrt{2}tx)=\sqrt{2}x *_{\tau}{\rm Exp}_{*_\tau} (\sqrt{2}tx), 
%\]
%we easily have 
%\[
%\frac{\tau}{\sqrt{2}}H'_n(x)+\sqrt{2}x H_n(x,\tau)=H_{n+1}(x,\tau)
%, \qquad (n=0,1,2\cdots).
%\]

As can be seen from these computations, the new product $*_{\stkap}$ allows for much simpler manipulation of various expressions than the ordinary commutative product $\cdot$.
%The new product $*_{\stkap}$ often makes 
%it easier to handle than using 
%ordinary product $\cdot$, 
%as demonstrated in this example. 
%We study more examples in section ???.
%\subsection{Intertwiner}
%As for augmented star product 
%$*_{\Gamma_1}, *_{\Gamma_2}$, 

\medskip

We can also define star Hermitian polynomials in several variables\footnote{By the same way, we can define theta functions in several variables. cf. \cite{miyazaki2026-1, mumford}}. 
First recall that 
for any functions $f,g$ and arbitrary complex $N\times N$-matrix
$\Gamma\in M_N({\mathbb C})$, {augmented star product} $*_{\Gamma}$ is defined as follows:
\[
f*_{\Gamma}g=fg+\sum_{\ell=1}^{\infty} \frac{1}{\ell!}\left(\frac{i\hbar}{2}\right)^{\ell}C_{\ell}^{\Gamma}(f,g), 
\]
where the bi-differential operators $C_{\ell}^{\Gamma}$ are defined by 
\[
C_{\ell}^{\Gamma}(f,g)
:=\sum_{
\tiny
\begin{array}{c}
j_r,k_s=1\\
(r,s=1,\ldots, \ell)
\end{array}
}^{N} \Gamma^{j_1k_1}\cdots
\Gamma^{j_{\ell}k_{\ell}}\partial_{j_1\ldots j_{\ell}}f \cdot \partial_{k_1\ldots k_{\ell}}g.  
\]

Using the formula in Proposition \ref{maillard-deg1}, i.e.  
\[
{\rm Exp}_*\frac{t}{i\hbar}({}^t {\bf b}{\bf x}+c) =\exp\left(\frac{t^2}{8i\hbar}{}^t{\bf b}(\Gamma + {}^t\Gamma){\bf b}  \right)
\times
\exp\left(\frac{t}{i\hbar}({}^t {\bf b}{\bf x}+c)\right)
\]
we introduce  

\begin{defn}
Set the ordering matrix 
\[
\displaystyle \Gamma=2\pi i K~(K\in Sym_{\mathbb C}(n)), 
~~\hbar =\frac{-1}{2\pi}, 
~~t=1, {\bf t}=\pi i {\bf b}, 
~~{\bf z}={\bf x}. 
\] 
Then we define $H_{\alpha}({\bf z},K)~(\alpha\mbox{:multi-index})$ as follows: 
\[
\sum_{\alpha\in {\mathbb Z}_{\geq 0}^n } 
 \frac{1}{\alpha !}H_{\alpha}({\bf z},K) {\bf t}^{\alpha}
:={\rm Exp}_{*_{2\pi i K}} \frac{t}{i\hbar}({}^t{\bf b x})
={\rm Exp}_{*_{2\pi i K}} (2{}^t{\bf t z})
=\exp({}^t{\bf t}K{\bf t}+2{\bf t z}). 
\]
Thus we have 
\[
H_{\alpha}({\bf z},K)
=
\partial_{\bf t}^{\alpha} \exp({}^t{\bf t}K{\bf t}+2{\bf t z}) \Bigr|_{{\bf t =0}}. 
\]
\end{defn}

\subsection{Augmented star Sonine (associated Laguerre) polynomials}
As an application of augmented star exponential of quadratic form, we introduce augmented star Sonine (associated Laguerre) polynomials. 
Recall that for any functions $f(z), g(z)$ of one variable $z$ and a complex number $\kappa$,
we define
\[
f(z)*_{\kappa} g(z):=f(z)\exp\Bigr\{ \frac{1}{2}\overleftarrow{\partial_z}\kappa\overrightarrow{\partial_z}  \Bigr\}g(z). 
\] 
%where we used $\tau$ instead of $\kappa$. 
This is the case of $1\times 1$ ordering matrix with complex coefficient $\displaystyle \Gamma=\Lambda +K=0+\kappa$ and $\hbar=-i$. 
\par\medskip\noindent
%{\bf Note that in what follows, we use $\startau$ 
%istead of $*_{2\pi i\tau}$ for short. }
Using the augmented star product $*_{\kappa}$, we consider the augmented star exponential of a quadratic form 
$\displaystyle z_{\stkap}^2=z^2+\frac{\kappa}{2}$.  
In this case, defining equation of the augmented star product is 
\begin{equation}
\frac{d}{dt}F(t,z)=w_{\stkap}^2 \stkap F(t,z), \quad 
F(0,z)=1
\end{equation}
with 
$F(t,z)=f(t)e^{g(t)w^2}$. 
That is 
\begin{equation}
\left\{
\begin{array}{lcl}
\vspace{2mm}
\displaystyle \frac{d}{dt}g(t)&=&(1+\kappa g(t))^2, \\
\vspace{2mm}
\displaystyle \frac{d}{dt}f(t)&=&\displaystyle \frac{1}{2}(\kappa^2 g(t) +\kappa)f(t), 
\end{array}
\right.
\end{equation}
with 
$g(0)=0,~f(t)=1$. 
Thus, we see 
\begin{equation}\label{w-square}
\displaystyle 
{\rm Exp}_{\stkap}(t z_{\stkap}^2)=
\frac{1}{\sqrt{1-t \kappa }}  
e^{\frac{t}{1- t \kappa } z^2} . 
\end{equation} 
Here we recall the definiton of generating function 
Sonine ((associated) Laguerre) polynomials $L_n^{(\frac{-1}{2})}$. 
\[
\frac{1}{ (1-t)^{\alpha+1} } e^{-\frac{tx}{1-t}}
=
\sum_{n=0}^{\infty}L_n^{(\alpha)}(x)t^n, ~~(|t|<1,={\rm Re}\alpha>-1). 
\]
If $\displaystyle \alpha=\frac{-1}{2}$, this is the $\kappa=-1$ expression of $e_{\stkap}^{-tw_{\stkap}^2} $: 
\[
\displaystyle 
{\rm Exp}_{\kappa=-1}(-t z_{*_{\kappa=-1} }^2)=
\frac{1}{\sqrt{ 1-t } }   
e^{\frac{-t}{1- t } z^2} . 
\]
Keeping this in mind, we define augmented star Sonine  polynomials $L_n(z^2,\kappa)$ by 
\begin{definition}
By \eqref{w-square}
we define 
\begin{equation}
{\rm Exp}_{\kappa}(t z_{*_{\kappa} }^2)
=
\frac{1}{\sqrt{ 1-t\kappa } }   
e^{\frac{t}{1- t\kappa } z^2} 
=\sum_{n=0}^{\infty} \frac{1}{n!} L_n(z^2,\kappa) t^n.
\end{equation}
Hence 
\[
L_n(z^2,\kappa)=\frac{d^n}{dt^n}\Bigr|_{t=0}\frac{1}{\sqrt{ 1-t\kappa } }   
e^{\frac{t}{1- t\kappa } z^2}. 
\]
\end{definition}

\section{Applications to quantum physics}
%\section{Concluding remarks} 
In this section, we give applications of augmented star products, augmented star exponentials and equivalence operators for several ordering matrices $\Gamma$. 
For this section, we strongly recommend consulting papers \cite{bffls, BW, dt, montiel, RW, zachosetal}.

\subsection{Husimi quantization and Husimi star product} % (Belchev, Robbins and Walton)}

Recall that 
as for a complex symmetric matrix $K$, we have 
\[
e^{\frac{i\hbar}{4}K^{ij}{\overrightarrow{\partial}_i\overrightarrow{\partial}_j}}(f*_\Gamma g)
=\Bigr(e^{\frac{i\hbar}{4}K^{ij}{\overrightarrow{\partial}_i\overrightarrow{\partial}_j}}(f)\Bigr)*_{\Gamma+K}\Bigr( e^{\frac{i\hbar}{4}K^{ij}{\overrightarrow{\partial}_i\overrightarrow{\partial}_j}}(g)\Bigr)
\]
That is, $T^K=e^{\frac{i\hbar}{4}K^{ij} 
{\overrightarrow{\partial}_i\overrightarrow{\partial}_j}}$ 
is an equivalence operator 
from an augmented star product $*_\Gamma$ 
to another augmented star product $*_{\Gamma+K}$. 
%\begin{defn}
%\[
%T^K:=
%e^{\frac{i\hbar}{4}K^{ij}{\overrightarrow{\partial}%_i\overrightarrow{\partial}_j}} . 
%\]
%\end{defn}
Using the equivalence operator above, 
we can unify several examples of quantization. 

As in Robbins and Walton \cite{RW}, 
it is known that a distribution in phase-space, $f(q, p)$,  coarse grained as follows:
\begin{equation}\label{grained}
\begin{array}{lcl}
\vspace{2mm}
&&
\displaystyle
\frac{1}{\pi \eta} \int dq' dp' f(q',p')
\exp\left\{
-\frac{1}{\eta}
\left[
\frac{(q-q')^2}{\sigma^2}+\sigma^2(p-p')^2
\right]
\right\}
\\
\vspace{2mm}
&=&
\displaystyle
\exp
\left[
\frac{\eta}{4}
\left(
\sigma^2\partial_q^2+\frac{1}{\sigma^2}\partial_p^2
\right)
\right]
f(q,p) , 
\end{array}
\end{equation}  
where $\eta$ is a classical coarse-grained scale parameter, independent of $\hbar$, and $\sigma$ is a squeezing parameter. 
When $\eta=\hbar$, and $f$ is the Wigner function $W$, the expressions in \eqref{grained} equal to the original Husimi quasi-probabilit distribution. 
Set  
\[
\displaystyle 
\Gamma=\Lambda_0=
\left[
\begin{array}{cc}
0 & 1 \\
-1 & 0
\end{array}
\right], 
\qquad 
K_{\rm Husimi}
=\frac{\eta}{i\hbar}
\left[
\begin{array}{cc}
\sigma^2 & 0 \\
0 & \displaystyle \frac{1}{\sigma^2}
\end{array}
\right]. 
\]  
Then, as for the Moyal star product $*_M=*_{\Lambda_0}$, 
using the equivalence operator 
$T^{K_{\rm Husimi}}$, we have 
\begin{equation}
T^{K_{\rm Husimi}}(*_M)=*_{\Lambda_0+K_{\rm Husimi}}=
\exp\left[
\frac{i\hbar}{2}
(\overleftarrow{\partial}_q\overrightarrow{\partial}_p
-\overleftarrow{\partial}_p\overrightarrow{\partial}_q)
+\frac{\eta}{2}
(\sigma^2\overleftarrow{\partial}_q\overrightarrow{\partial}_q
+\frac{1}{\sigma^2}\overleftarrow{\partial}_p\overrightarrow{\partial}_p)
\right]
\end{equation}
and this gives an associative product. 
This is 
%\[
%I^{K_{\rm Husimi}}(*_M)=*_{\Lambda_0+K_{\rm %Husimi}}, 
%\]
called {\it Husimi star product} in Robbins and Walton \cite{RW}. 

\medskip

We can also construct a one-parameter family of augmented star products as follows:
\[
\{ T^{sK_{\rm Husimi}}(*_M)\}_{s\in[0,1]}=\{*_{\Lambda_0+sK_{\rm Husimi}}
\}_{s\in[0,1]}.
\]
Furthermore, for this family and any quadratic form 
$Q={}^txAx$ ($x={}^t (q,p)$), we have 
\begin{equation}
\begin{array}{lcl}
\vspace{2mm}
&&\displaystyle
{\rm Exp}_{*_{\Lambda_0+sK_{\rm Husimi}}}\left(\frac{t}{i\hbar}Q\right)
%\\
%\vspace{2mm}
%&=&
%\displaystyle
%\frac{1}{f(t)}\exp(i[{}^txg(t)x]) 
\\
\vspace{2mm}
&=&
\displaystyle 
\det{}^{-\frac{1}{2} }
\left(
\frac{I-{\Lambda_0sK_{\rm Husimi}}+e^{2A\Lambda_0 t} (I+{\Lambda_0 sK_{\rm Husimi}}) }{2}
\right)
\\
\vspace{2mm}
&&
\displaystyle 
\qquad\qquad
\times
\exp
\left(
i {}^tx 
\frac{-1}{\hbar}
\Bigr[
\frac{I}{I-{\Lambda_0 sK_{\rm Husimi}}+e^{2A\Lambda_0 t}(I+{\Lambda_0 sK_{\rm Husimi}})}(I-e^{2A\Lambda_0 t})\Lambda_0
%(I-e^{2AJt})
%\frac{I}{I-JK+(I+JK)e^{2AJt}}J
\Bigr]
x
\right). 
\end{array}
\end{equation} 
%Note that any other one-parameter family 
%with the same end-points, as long as it is %homotopic outside the zero set of 
%$\det^{\frac{1}{2}}$, preserves 
%the same topological information of this system. 

\medskip\noindent
Note that for any dimensional symplectic space, we can extend the above results.

\subsection{Harmonic oscillators and damped quantizations and damped star products}
In this subsection, we are concerned with harmonic oscillators and damped harmonic oscillators. 
Consider the harmonic oscillator, with Hamiltonian function
\[
H(q,p)=\frac{1}{2m}p^2+\frac{1}{2}m\omega^2q^2
.
\]
%P=\partial_q\partial_p-\partial_q\partial_q
Replacing the Poisson bivector
$\overleftarrow{\partial}_q\overrightarrow{\partial}_p-\overleftarrow{\partial}_q\overrightarrow{\partial}_q$ with  
\begin{equation}
\overleftarrow{\partial}_q\overrightarrow{\partial}_p-\overleftarrow{\partial}_p\overrightarrow{\partial}_q
-2\gamma m
\overleftarrow{\partial}_p\overrightarrow{\partial}_p, 
\end{equation}
the canonical equations of a damped harmonic oscillator changed into:
\[
\begin{array}{lcl}
\vspace{2mm}
\dot{q}
&=&
\displaystyle
q (\overleftarrow{\partial}_q\overrightarrow{\partial}_p-\overleftarrow{\partial}_p\overrightarrow{\partial}_q
-2\gamma m
\overleftarrow{\partial}_p\overrightarrow{\partial}_p) H(q,p) =\frac{p}{m}, 
\\
\vspace{2mm}
\dot{p}
&=&
\displaystyle
p
(\overleftarrow{\partial}_q\overrightarrow{\partial}_p-\overleftarrow{\partial}_p\overrightarrow{\partial}_q
-2\gamma m
\overleftarrow{\partial}_p\overrightarrow{\partial}_p)
 H(q,p) =-m\omega^2q -2\gamma p. 
\end{array}
\]
Hence the equation of motion of a damped harmonic oscillator results:
\[
\ddot{q}=-\omega^2q-2\gamma \dot{q} ,
\]
with $\gamma$ as the damping parameter. 

Set  
\[
\displaystyle 
\Gamma=\Lambda_0, \qquad 
K_{\rm DH}
=
\left[
\begin{array}{cc}
-2\gamma m & 0 \\
0 & \displaystyle 0
\end{array}
\right]. 
\]  
Then as for the Moyal star product $*_M=*_{\Lambda_0}$, we have 
\begin{equation}\label{dh-star-exp}
T^{K_{\rm DH}}(*_M)=*_{\Lambda_0+K_{\rm DH}}=
\exp\left[
\frac{i\hbar}{2}
(\overleftarrow{\partial}_q\overrightarrow{\partial}_p
-\overleftarrow{\partial}_p\overrightarrow{\partial}_q)
+\frac{-2i\hbar \gamma m}{2}
(\overleftarrow{\partial}_p\overrightarrow{\partial}_p)
\right]
\end{equation}
and this gives an associative product. 
This is 
%\[
%I^{K_{\rm Husimi}}(*_M)=*_{\Lambda_0+K_{\rm %Husimi}}, 
%\]
called {\it damped star product}. See \cite{dt}, \cite{montiel}  and \cite{RW}. 
Note that \eqref{dh-star-exp} obtained by our  equivalence  operator coincides with (27) in Dito and Turrubiates \cite{dt}. 

\medskip

We can also construct a one-parameter family of augmented star products as follows:
\[
\{ T^{sK_{\rm DH}}(*_M)\}_{s\in[0,1]}=\{*_{\Lambda_0+sK_{\rm DH}}
\}_{s\in[0,1]}.
\]
Furthermore, for this family and any quadratic form 
$Q={}^txAx$ ($x={}^t (q,p)$), we have 
\begin{equation}\label{dh-star-exponential}
\begin{array}{lcl}
\vspace{2mm}
&&\displaystyle
{\rm Exp}_{*_{\Lambda_0+sK_{\rm DH}}}
\left(\frac{t}{i\hbar}Q\right)
%\\
%\vspace{2mm}
%&=&
%\displaystyle
%\frac{1}{f(t)}\exp(i[{}^txg(t)x]) 
\\
\vspace{2mm}
&=&
\displaystyle 
\det{}^{-\frac{1}{2} }
\left(
\frac{I-{\Lambda_0sK_{\rm DH}}+e^{2A\Lambda_0 t} (I+{\Lambda_0 sK_{\rm DH}}) }{2}
\right)
\\
\vspace{2mm}
&&
\displaystyle 
\qquad\qquad
\times
\exp
\left(
i {}^tx 
\frac{-1}{\hbar}
\Bigr[
\frac{I}{I-{\Lambda_0 sK_{\rm DH}}+e^{2A\Lambda_0 t}(I+{\Lambda_0 sK_{\rm DH}})}(I-e^{2A\Lambda_0 t})\Lambda_0
%(I-e^{2AJt})
%\frac{I}{I-JK+(I+JK)e^{2AJt}}J
\Bigr]
x
\right). 
\end{array}
\end{equation} 

When a quadratic form is 
\[
Q=H_{\rm Harm}=\frac{1}{2m}p^2+\frac{1}{2}m\omega^2q^2,
\] 
a direct computation shows that 
\begin{equation}\label{dh-harm-star-exp}
\begin{array}{lcl}
\vspace{2mm}
\eqref{dh-star-exponential}
&=&
\displaystyle 
\frac{\exp(\gamma t/2)}{\cos(\omega t/2) \bigr\{1+\frac{2\gamma}{\omega}\tan(\omega t/2) \bigr\}^{1/2}}\\
\vspace{2mm}
&&
\qquad\qquad
\displaystyle 
\times
\exp\left[
\frac{-i}{\hbar\omega}
\tan(\omega t/2)
\Bigr( m\omega^2q^2 + \frac{p^2}{m\{1+\frac{2\gamma}{\omega}\tan(\omega t/2) \} } \Bigr)
\right]. 
\end{array}
\end{equation}

\subsection{The Feynman-Kac formula and ground state energy in deformation quantization} 
In this subsection we compare quantum states of harmonic oscillator and damped harmonic oscillator. 

Using deformation quantization formulation of the Feynman-Kac formula\footnote{See \cite{funaki, ks} for the original Feynman-Kac formula.} 
for the ground state energy $E_0$ (\cite{montiel}):
\begin{equation}
E_0=-\lim_{\tau \to \infty}
\frac{\hbar}{\tau}\log
\left[
\frac{1}{2 \pi \hbar} 
\int_{\R^2} {\rm Exp}_{*} \Bigr(-\frac{\tau}{\hbar}H\Bigr)dqdp
\right], 
\end{equation}
with a family 
$\{*=T^{sK_{\rm DH}}({*_{\Lambda_0}})
=*_{\Lambda_0 + s K_{\rm DH} } \}_{s\in I}$, 
$H=H_{\rm Harm}$ and a Wick rotation $\tau=it$, we obtain 
\begin{equation}
\begin{array}{lcl}
\vspace{2mm}
E_0(\gamma, s)
&=&
\displaystyle
-\lim_{\tau \to \infty}
\frac{\hbar}{\tau}\log
\Bigr[
\frac{1}{2 \pi \hbar} 
\int_{\R^2} {\rm Exp}_{T^{sK_{\rm DH}}({*_{\Lambda_0}})}
%{*_{ {\Lambda_0 + K_{\rm DH}}}
\Bigr( -\frac{\tau}{\hbar} H_{\rm Harm} \Bigr)dqdp
\Bigr]
\\
\vspace{2mm}
&\stackrel{\eqref{dh-harm-star-exp}}{=}&
\displaystyle
-\lim_{\tau \to \infty}
\frac{\hbar}{\tau}\log
\Bigr[
\frac{1}{2}e^{-i s \gamma \tau /2}
{\rm cosech} (\omega \tau /2)
\Bigr]
\\
\vspace{2mm}
&=&
\displaystyle
\hbar\omega(1/2+i s \gamma/2\omega)
, 
\end{array}
\end{equation} 
where 
$\displaystyle 
{\rm cosech}(x)=2 \frac{e^x}{e^{2x}-1}. 
$
Note that $s \to 0$ or $\gamma \to 0$ then 
$\displaystyle E_0(\gamma,s)$ %uniformly
converges $\displaystyle \frac{\hbar\omega}{2}$ which is the ground state energy of harmonic oscillator. 

\medskip

Furthermore, using the formula of augmented 
star exponential for general quadratic form 
on the phase space $\R^2$, we can extend the above result in the following way \cite{montiel}. 
Let $H_{\rm quad}(q,p)$ be a general quadratic Hamiltonian of the form 
\[
H_{\rm quad}(q,p)=aq^2+2cqp+bp^2, 
\] 
where $a,b,c \in \R.$
In this case, using a Wick rotation, 
the augmented star exponential takes the form 
\[
{\rm Exp}_{*_{\Lambda_0}} 
\left(
-\frac{\tau}{\hbar}H_{\rm quad}(q,p)
\right)
=
\frac{1}{\cosh(\sqrt{ab-c^2} \tau) }
e^{-\frac{iH_{\rm quad}(q,p)}{\hbar \sqrt{ab-c^2} } \tanh\sqrt{ab-c^2}\tau}. 
\]
According to the Feynman-Kac formula, 
we have 
\[
\frac{1}{2\pi\hbar} 
\int_{\R^2} {\rm Exp}_{*_{\Lambda_0}} 
\left(
-\frac{\tau}{\hbar}
H_{\rm quad}(q,p)
\right)
dqdp
=
\frac{1}{2}{\rm cosech}(\sqrt{ab-c^2}\tau), 
\]
and then 
\[
E_0 
=
-\lim_{\tau \to \infty}
\frac{\hbar}{\tau}
\log\frac{1}{2}{\rm cosech}(\sqrt{ab-c^2}\tau)
=\hbar\sqrt{ab-c^2}. 
\]

\section{Concluding Remarks}
As mentioned above, there are so many applications of augmented star product and augmented star exponentials to special functions and physics.

In the case of 
the real symplectic group $Sp_{2n}(\R)$, 
for matrices $g_1,~g_2$ 
close to the identity matrix, we have 
\[
\begin{array}{lcl}
\vspace{3mm}
&&
\displaystyle
{\rm Exp}_{*_\Gamma}\left(
\frac{1}{i\hbar}{}^tx \bigr( \frac{1}{2}(\log g_1)\Lambda^{-1} \bigr) x \right)
{*_\Gamma}~
\displaystyle
{\rm Exp}_{*_\Gamma}\left(
\frac{1}{i\hbar}{}^tx \bigr( \frac{1}{2}(\log g_2)\Lambda^{-1} \bigr) x  \right)
\\
\vspace{3mm}
&=&
\displaystyle 
\det{}^{-\frac{1}{2} }
\left(
\frac{I-\Lambda K+(g_1g_2) (I+\Lambda K) }{2}
\right)
%\\
%\vspace{4mm}
%&&
\displaystyle 
%\qquad\qquad
\cdot
\exp
\left(
 {}^tx 
\frac{1}{ i \hbar}
\Bigr[
\frac{I}{I-\Lambda K+(g_1g_2)
(I+\Lambda K)}(I-(g_1g_2))\Lambda 
\Bigr]
x
\right). 
\end{array}
\] 
Note that 
for product formulas involving points $g_1, g_2$ 
away from the identity matrix, 
it is necessary to introduce 
a correction factor using the Maslov triple product (\cite{miyazaki2026-1, nomura}).

%–{•¶'±'±'Ü'Å

\end{document}